\documentclass[12pt, journal,  onecolumn]{IEEEtran}
\usepackage[margin=1in]{geometry}
\usepackage{setspace}
\usepackage{color}
\usepackage{enumitem}
\usepackage{graphicx,tabularx,array,amsmath,amsthm,thmtools,subfigure}
\usepackage{hyperref}

\usepackage{mathtools}

\usepackage{amsfonts}
\usepackage{bm}
\usepackage{bbm}
\usepackage{makecell}
\usepackage{multirow}
 \usepackage{amssymb}
\usepackage{txfonts}
\usepackage[T1]{fontenc}
\usepackage{tikz}
\usepackage[scr=dutchcal]{mathalfa}
\usepackage{mathtools}
\usepackage{cite}
\let\mathscr\mathbscr
\usepackage{cuted}
\usepackage[]{algorithm2e}
\usepackage{lipsum,geometry}

\newtheorem{Theorem}{Theorem}
\newtheorem{Proposition}{Proposition}

\newtheorem{Example}{Example}
\newtheorem{Remark}{Remark}
\newtheorem{Definition}{Definition}

\ifCLASSINFOpdf
\else
\fi
\begin{document}
%
\newgeometry{top=1in,bottom=0.75in,right=0.75in,left=0.75in}

\title{Fundamental Privacy Limits in Bipartite Networks under Active Attacks}

\author{

\IEEEauthorblockN{ Mahshad Shariatnasab$^\dagger$\thanks{This work was supported in part by NSF grant CCF-1815821, CNS-1619129, and  ND EPSCoR grant FAR0033968. This work was presented in part at the Annual Allerton Conference on Communication, Control, and Computing (Allerton) 2017, and the International Symposium on Information Theory (ISIT) 2018.}, Farhad Shirani$^\dagger$, Elza Erkip$^\ddagger$}
\\\IEEEauthorblockA{$^{ \dagger}$North Dakota State University, $^{^\ddagger}$New York University
\\Email: $\{$mahshad.shariatnasab, f.shiranichaharsoogh$\}$@ndsu.edu, 
elza@nyu.edu
}
}


%


\maketitle

 \begin{abstract}
This work considers active deanonymization of bipartite networks. The scenario arises naturally in evaluating privacy in various applications such as social networks, mobility networks, and medical databases. 
For instance, in active deanonymization of social networks, an anonymous victim is targeted by an attacker (e.g. the victim visits the attacker's website), and the attacker queries her group memberships (e.g. by querying the browser history) to deanonymize her.
In this work, the fundamental limits of privacy, in terms of  the minimum number of queries necessary for deanonymization, is investigated. A stochastic model is considered, where i) the bipartite network of  group memberships is generated randomly, ii) the attacker has partial prior knowledge of the group memberships, and iii) it receives noisy responses to its real-time queries. The bipartite network is generated based on linear and sublinear preferential attachment, and the stochastic block model. The victim's identity is chosen randomly based on a distribution modeling the users' risk of being the victim (e.g. probability of visiting the website). An attack algorithm is proposed which builds upon techniques from communication with feedback, and its performance, in terms of expected number of queries, is analyzed. Simulation results are provided to verify the theoretical derivations.


\end{abstract}


%
\IEEEpeerreviewmaketitle

\section{Introduction}
As tracking technologies --- both online and in the real-world ---
become more sophisticated and pervasive, there is a critical need to understand and quantify  the resulting privacy risk. For instance, on the
internet, users reasonably expect their online identities and web browsing activities to remain private. Unfortunately, this is far from the case in practice; in reality, users are constantly tracked on the internet. 
Often this is for benign, if somewhat disconcerting, reasons --- for instance, websites track users to serve them with targeted digital advertisements \cite{estrada2017online,razaghpanah2018apps}.
More disturbingly, web tracking can be used to stifle individuals'
free speech rights, or target vulnerable minority groups \cite{mavriki2017using}. 
Furthermore, in wireless applications, the location-based services
offered by
mobile devices, such as smart phones and autonomous vehicles, can cause significant privacy threats to users, since
the time series of locations can be statistically
matched to prior user behavior and lead to  identification
and tracking \cite{takbiri2017limits,montazeri2017achieving,takbiri2019asymptotic,de2013unique,blondel2015survey}.
As a result, there is an urgent need to understand and quantify users' privacy risk, that is, what is the likelihood that users on  can be uniquely identified using their \emph{fingerprints}? In this work, we study the fundamental limits of privacy in bipartite networks under active attacks. These networks arise naturally in modeling social network group memberships \cite{kruegel,fire2014online,su2017anonymizing}, medical databases \cite{domingo2016database}, and 
wireless mobility data \cite{montazeri2017achieving,takbiri2019asymptotic,de2013unique,blondel2015survey}, among others. 

The browser social network deanonymization attack developed 
by Wondracek et al.~\cite{kruegel} is a good representative of practical active bipartite network deanonymization (ABND)  attacks in the literature,
where the attacker runs a malicious website and seeks to  deanonymize users who visit the website (see Figure~\ref{fig:fingerprinting}).
To this end, the attacker first 
uses a web scraper to scrape the group memberships of users. This serves as the attacker's scanned bipartite graph,
$\mathcal{G}_s$, capturing the social network group memberships. 
Note that the scanned graph might be different from the ground-truth because of users privacy settings that act as a source of noise. When an unknown user (the victim) visits the attacker's website, the attacker queries social network group memberships to find the victim's identity. This is done by using browser history sniffing \cite{gulmezoglu2017perfweb, smith2018browser,shusterman2019robust} to ask questions of the form ``is the webpage of social network group `$r_{j}$' in the victim's browser history?" If yes, the attacker assumes that the victim is a member of the social network group $r_{j}$, and if no then the attacker assumes the victim is not a member of $r_{j}$. Of course, a user might be a member of a group they have not visited, or conversely, might not be a member of a group they have visited; consequently, the attacker's measurement is noisy. The attacker repeats this query for all social network groups in a pre-determined set to obtain the unknown victim's partial fingerprint. 
By matching the partial fingerprint of query responses to
the scanned fingerprints in the scanned graph the victim is deanonymized. 
In \cite{kruegel}, 
this simple deanonymization strategy is evaluated by using it to find the identities of the users in the Xing social network. It is shown that over 42\% of the users who are members of at least one group on Xing (more than 5.7 million users) can be deanonymized successfully using the algorithm.
Although effective, 
Wondracek et al.'s attack does not answer 
fundamental questions about the {optimal} number and type of group memberships to query, and the order in which to issue queries.
Other fingerprinting attacks proposed in literature~\cite{ finger1,finger2,finger3,finger4,finger5} have also adopted similar 
ad-hoc approaches without theoretical guarantees or analyses.

A user's fingerprint is the set of group memberships that reflect the user's activities and habits, e.g. websites the user has visited and social network groups that a user is a member of ~\cite{attribute1,attribute2}, characteristics of the user's web browser (e.g. font size) ~\cite{browser1}, and physical device features ~\cite{hardware1}. Fingerprinting based deanonymization attacks build on the empirical observation that, for a large enough set of group memberships, a user's fingerprints are unique.
The challenge, from an attacker's standpoint, is that the victim's fingerprints may not be accurately or easily available; i.e., fingerprints may be noisy and the attacker may have to actively query the victim's group memberships, one group at a time, to measure their fingerprint. However, an attacker may only be able to issue a limited number of queries to the victim's device. Our objective is to provide a rigorous mathematical formulation along with theoretical privacy guarantees for the ABND scenario.

\begin{figure}
\begin{center}
\includegraphics[width=0.6\textwidth]{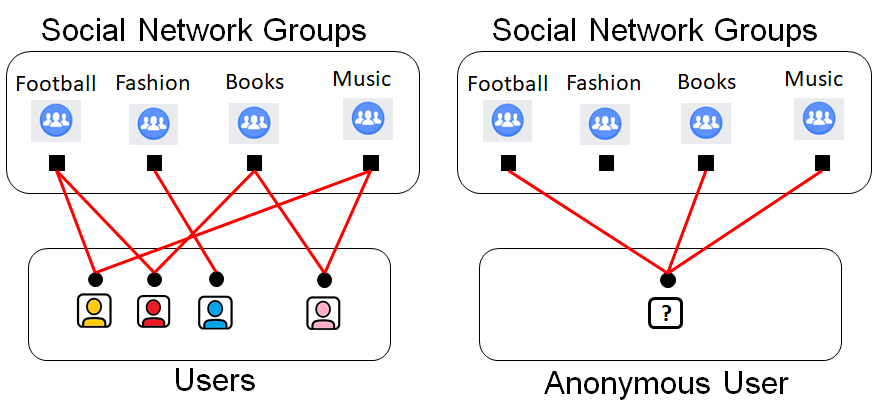}
\caption{(Left) An example of a group membership bipartite graph. (Right) An anonymous user (victim) is to be deanonymized based on partial fingerprints.}
\label{fig:fingerprinting}
\end{center} 
\end{figure} 

In ~\cite{shirani2017information}, we proposed a mathematical formulation for the ABND problem and introduced a typicality-based strategy by making analogies to the problem of  channel coding in information theory, and quantified the amount of information the attacker obtains from each query. We showed that under the assumption that users are equally likely to visit the attacker's website, the total number of queries required for deanonymization grows logarithmically in the number of users. Furthermore, the coefficient of the logarithm is inversely proportional to the mutual information between the random variables corresponding to the scanned graph elements and query responses. In \cite{shirani2018optimal}, we considered a general distribution, as opposed the a uniform one,  on the victim's index among the social network users.
This is based on the intuition that more active users would be more likely to visit an attacker's website, resulting in a non-uniform distribution on the victim index. We used techniques from communication over channels with feedback with non-uniform message sets, to propose attack strategies and derive theoretical performance guarantees.

In ~\cite{shirani2017information,shirani2018optimal}, we  considered random bipartite network models in which the edges are independent and identically distributed. However, many bipartite networks of interest, such as social networks \cite{capocci2006preferential,newman2001clustering}, networks in cell biology \cite{albert2005scale}, mobility networks \cite{borrel2005preferential}, and collaboration networks \cite{van2009random,bloznelis2015random} resemble graphs which are generated based on a growing model that grows in accordance to the preferential attachment (PA) rule, first proposed by Simon \cite{simon1957models} and rediscovered by Barb\'{a}si and Albert \cite{barabasi1999emergence}. In this model, edges are added to the graph iteratively, where at each step, a set of edges are added to the graph randomly such that vertices which have a higher degree are more likely to attract more new connections. In addition to the PA model, another random bipartite graph generation model of interest is the stochastic block (SB) model, where groups are divided into communities, and community memberships of groups affects their likelihood of attracting new users \cite{florescu2016spectral,yen2020community}.
In this work, we propose a general formulation for the ABND problem,  where the bipartite graph random generation model encompasses the PA and SB models, and the scan and query noise models capture the users' different privacy settings and device specifications. We further propose several
information-threshold-based deanonymization strategies which build upon the channel coding and hypothesis testing methods studied in ~\cite{burnashev1976data,naghshvar2013active} to devise deanonymization attacks, and analyze their performance in terms of expected number of queries for successful deanonymization. Our main contributions are summarized below:

{\begin{itemize}[leftmargin=*]
\item{We build upon the ideas  in \cite{shirani2017information,shirani2018optimal} to develop a general mathematical formulation of the ABND problem which encompasses the network generation models such as PA and SB models, and allows for scan and query noises with general distributions. These distributions capture the users' various privacy preferences and device specifications.}
\item{We study the degree distribution and statistical properties of the graph under the proposed generation model. We prove that under certain sparsity conditions on the graph edges, the correlation among the user fingerprints is `weak' and the fingerprint vector's distribution is well-approximated by a product distribution. These derivations may be of independent interest in the study of bipartite networks.}
\item{We propose information-threshold-based attack strategies and derive theoretical guarantees for theirs success. Roughly speaking, in the proposed strategies, the attacker queries the selected victim's group memberships sequentially and calculates the amount of information obtained, i.e. the amount of uncertainty regarding each user index based on previous query responses. The attack ends when the uncertainty is lower than a given threshold for one of the user indices. The strategy reduces to the one in \cite{shirani2018optimal} if the graph edges are assumed to be independent and equally probable, which was proved to be optimal  in terms of expected number of queries necessary for successful deanonymization for asymptotically large networks.  }
\item{We simulate the performance of the proposed strategies both for synthesized as well as real-world networks, and compare the results with our analytical derivations.
}
\end{itemize}}

The rest of the paper is organized as follows:  Section \ref{sec:not} describes the notation. In Section \ref{sec:form}, we provide the problem formulation. In Section \ref{sec:mem}, we study the degree distribution and other statistical properties of the graph. In Section \ref{sec:ITS}, we propose the attack strategy and derive theoretical guarantees for its success. In Section \ref{sec:simul}, we provide simulation results to verify the theoretical derivations. Section \ref{sec:conc},  concludes the paper.

\section{Notation}\label{sec:not}
 We represent random variables by capital letters such as $X, U$ and their realizations by small letters such as $x, u$. Sets are denoted by calligraphic letters such as $\mathcal{X}, \mathcal{U}$. The set of natural numbers, and the real numbers are represented by $\mathbb{N}$, and $\mathbb{R}$ respectively. The random variable $\mathbbm{1}_{\mathcal{E}}$ is the indicator function of the event $\mathcal{E}$.
 The set of numbers $\{n,n+1,\cdots, m\}, n,m\in \mathbb{N}$ is represented by $[n,m]$. Furthermore, for the interval $[1,m]$, we sometimes use the shorthand notation $[m]$ for brevity. 
 For a given $n\in \mathbb{N}$, the $n$-length vector $(x_1,x_2,\hdots, x_n)$ is written as $x^n$. 

\section{Problem Formulation}
\label{sec:form}

In this section, we describe our mathematical formulation of the ABND scenario,  which generalizes the formulation provided in ~\cite{shirani2017information,shirani2018optimal}, and encompasses the statistical models for bipartite networks proposed in \cite{kunegis2013preferential,peruani2007emergence,peltomaki2006correlations}. To facilitate explanation, and provide justifications for the model assumptions, we describe the model by focusing on the scenario of deanonymizing social network users using the bipartite network of their group memberships. 
An ABND attack unfolds in two phases, a passive phase, and an active phase   \cite{wondracek2010practical, olejnik2012johnny,su2017anonymizing}. In the passive phase, the attacker acquires a noisy observation of the bipartite network of group memberships by scanning the whole social network. In the active phase, the attacker targets a specific victim (e.g. a user visiting the attacker's website), and uses browser history sniffing techniques to query the victim's group memberships. The attacker constructs a \emph{fingerprint} for the victim using the (noisy) query responses, and identifies the victim by comparing this fingerprint with the noisy scan of the bipartite graph acquired in the passive phase of the attack. 
As shown in Figure \ref{fig:overview}, the model consists of three components which are described in detail in the following sections: i) the \emph{ground-truth} $\mathcal{G}_0$ representing the `true' group memberships of users in the social network
(Section \ref{sec:ground_truth}), ii) the \emph{scanned graph} $\mathcal{G}_s$ which represents the attacker's prior knowledge of the ground-truth
(Section \ref{sec:scanned_graph}), and iii) the \emph{query responses}, represented by $\mathcal{G}_q$, which are acquired by the attacker by querying the victim in the active phase of the attack (Section \ref{sec:query_responses}). The objective is to design an attack strategy which determines the sequence of queries made by the attacker to deanonymize the victim, along with theoretical guarantees for its success (Section \ref{sec:strategy}).

\begin{figure}[!t]
\begin{center}
\includegraphics[width=0.7\columnwidth]{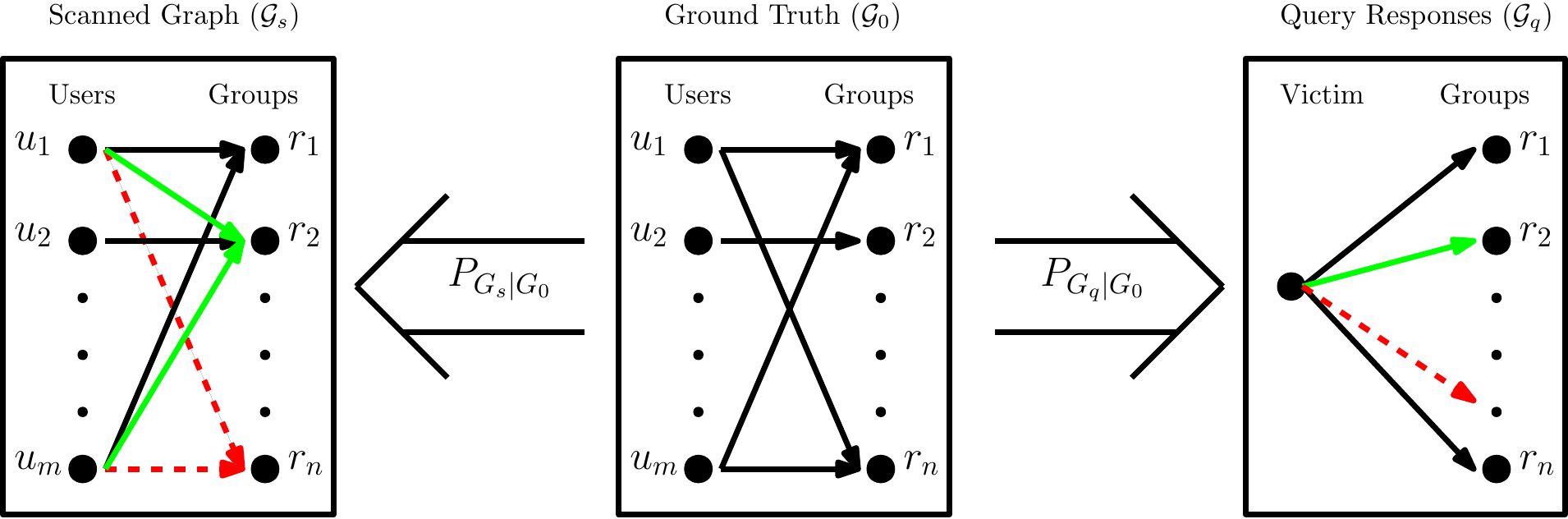}
\caption{Components of the ABND problem: i) the ground-truth characterized by the bipartite graph $\mathcal{G}_0$ and generated based on $P_{\mathcal{G}_0}$, ii) the scanned graph $\mathcal{G}_s$ generated based on $P_{\mathcal{G}_s|\mathcal{G}_0}$, and iii) the query responses $\mathcal{G}_q$ generated based on $P_{\mathcal{G}_q|\mathcal{G}_0}$. The black edges represent `true' group memberships, whereas green and dashed-red edges show additions and omissions, respectively, which may manifest due to noise in scanning the social network in passive phase of the attack, and noisy query responses in the active phase.}
\label{fig:overview}
\end{center}
\end{figure}

\subsection{The ground-truth}
\label{sec:ground_truth}
The collective set of group memberships in the social network are called the {\em ground-truth}. The ground-truth is represented by a bipartite graph. 

\begin{Definition}[\textbf{Bipartite Graph}]
A bipartite graph $\mathcal{G}=(\mathcal{V}_1,\mathcal{V}_2,{\mathcal{E}})$, is a graph with vertex set $\mathcal{V}_1\bigcup \mathcal{V}_2$ and edge set ${\mathcal{E}}\subseteq \{(v_i,v_j)|v_i\in \mathcal{V}_1, v_j\in \mathcal{V}_2\}$, where $\mathcal{V}_1\cap \mathcal{V}_2=\phi$. 
\end{Definition}

We consider a social network with user set $\mathcal{U}\triangleq \{u_1,$ $u_2,\cdots, u_m\}, m\in \mathbb{N}$, and group set $\mathcal{R}\triangleq \{r_1,$ $r_2,\cdots, r_n\}, n\in \mathbb{N}$. The ground-truth is characterized by a bipartite graph $\mathcal{G}_0= (\mathcal{U}, \mathcal{R}, \mathcal{E})$, where ($\mathcal{U}$,$\mathcal{R}$) partitions the vertex set, and the edge set $\mathcal{E}$  consists of all pairs $(u_k,r_j), k\in [m], j\in [n]$ for which user $u_k$ is a \textit{member} of the group $r_j$. 

\begin{Definition}[\textbf{Group Size}]
Let the set of users which are \textit{members} of the $j$th group $r_j, j\in [n]$ be denoted by $\mathcal{U}_j \triangleq \{u_{k_1}, u_{k_2}, \cdots , u_{k_{D_j}}\}, k_1,k_2,\cdots,k_{D_j}\in [m]$. Then, $D_j\triangleq |\mathcal{U}_j|$ is called the size of group $r_j$. 
\end{Definition}

\begin{Example}
 In the Facebook social network, $\mathcal{U}$ is the set of users and $\mathcal{R}$ includes the pages/ events/ groups/  applications on Facebook. Here, the groups under consideration are those whose member lists are publicly available. 
\end{Example}

Each user is assigned a fingerprint based on its group memberships. The fingerprint is a binary vector of indicator functions, indicating the membership of the user in each particular group. Alternatively, the user's fingerprint is the vector of indicator functions corresponding to the edges between the user and each of the groups. 

\begin{Definition}[\textbf{Fingerprint}]
\label{def:finger}
Consider the ground-truth bipartite graph  $\mathcal{G}_0=(\mathcal{U},\mathcal{R},{\mathcal{E}})$:
\begin{itemize}
\item{ For a user $u_k, k\in [m]$, the set $\mathcal{R}_k\triangleq \{r_j|(u_k,r_j)\in \mathcal{E}\}, k\in [m]$ is called the set of groups associated with $u_k$.}
\item{ The fingerprint of user $u_k, k\in [m]$ is the vector $({R}_{k,j})_{j\in [n]}\triangleq (R_{k,1},R_{k,2},\cdots,R_{k,n})$, where
\begin{align*}
R_{k,j}\triangleq
\begin{cases}
 1 \qquad &\text{if } r_j\in \mathcal{R}_k\\
0 &\text{otherwise}
\end{cases}, \qquad k\in [m], j \in [n].
\end{align*}}
\item{ The vector $R_{k,\mathcal{I}}\triangleq(R_{k,j})_{j\in \mathcal{I}}$ is called a partial fingerprint of $u_k, k\in [m]$, where $\mathcal{I}\subseteq [n]$.
}
\end{itemize}
\end{Definition}



We consider a stochastic model which is a generalization of those considered in prior works on active social network deanonymization \cite{shirani2017information, shirani2018optimal,su2017anonymizing}, and includes as a special case several statistical models such as  SB model, and PA model which have been used for bipartite networks such as social network group memberships, collaboration networks, authorship networks, and location networks \cite{kunegis2013preferential,peruani2007emergence,peltomaki2006correlations}. 

The ground-truth $\mathcal{G}_0$ is generated iteratively based on a `growing network' model as follows. Fix $\mu\in \mathbb{N}$, and define  $\Delta\triangleq\mu n$, where $n$ is the number of social network groups. The iterative process is initiated by considering a bipartite graph $(\mathcal{U}, \mathcal{R},\phi)$, which has no edges connecting its two sets of vertices. The ground-truth graph is generated in $\Delta$ iterative steps, where at each step a single edge is added to the graph, so that $|\mathcal{E}|=\Delta$ after the last iteration. As a result, the average group size is equal to $\frac{\Delta}{n}=\mu$.
For $t\in [\Delta]$, define  $\mathcal{G}_0(t)\triangleq (\mathcal{U},\mathcal{R}, \mathcal{E}(t))$ as the bipartite graph at step $t$. The group membership sets at step $t\in \Delta$ are denoted by $\mathcal{U}_j(t), j\in [n]$, and the group sizes are denoted by $D_{t,j}\triangleq |\mathcal{U}_j(t)|$. 
Building upon the idea of PA graph generation models --- where the likelihood that a given vertex connects to a new vertex is linearly related with the degree of that vertex ---
we assume that, at each step, groups attract new members in accordance with their \textit{popularity} at that step. To elaborate, we assume that each group $r_j, j\in [n]$ is assigned a {popularity value} $\tau_j(t)$ which captures its popularity at time $t$. The value of $\tau_j(t)$, which may depend on the size of group $r_j$ among other factors, affects the probability of $r_j$ attracting new members as described in the sequel.  In this work, we restrict to to the case where the value of $\tau_j(t), j\in [n],t\in [\Delta]$ depends only on the group size $D_{t,j}$ and an initial value $\tau_j(0)$.  The vector $\underline{\tau}(t)= (\tau_1(t),\tau_2(t),\cdots, \tau_r(t))$ represents the vector of group popularity values at time t. 
\\\textbf{Initiation:} Each group $r_j, j\in [n]$ is assigned an initial popularity value $\tau_j(0)>0$. The ground-truth graph is initiated as $\mathcal{G}_0(0)\triangleq (\mathcal{U},\mathcal{R}, \phi)$. So, the group membership sets are  $\mathcal{U}_j(0)= \phi, j\in [n]$ and $D_{0,j}=0, j\in [n]$. 
\\\textbf{Step t:}  At each step $t\in [\Delta]$, a group $r_{J_t}$ and a user $u_{K_t}$ are chosen as described next, and the corresponding edge $(u_{K_t}, r_{J_t})$ is added to the bipartite graph, i.e. $\mathcal{E}(t)=$ $\mathcal{E}(t-1)\cup \{(u_{K_t}, r_{J_t})\}$.  First, a group $r_{J_t}$ is chosen among the set of all groups $\mathcal{R}$ according to the probability distribution $\mathbf{P}(t)=(P_1(t),P_2(t),\cdots,$ $P_n(t))$ defined below:
\begin{align*}
    P_j(t)\triangleq\frac{\tau_{j}(t-1)}{\sum_{j'=1}^n \tau_{j'}(t-1)},
\end{align*}
 Next, a user $u_{K_t}$ is chosen randomly and uniformly from the set of users which are not members of $r_{J_t}$, i.e. $[m]- \mathcal{U}_{J_t}(t-1)$.  The edge $(u_{K_t},r_{J_t})$ is added to the edge set. 
The group popularity values are updated as follows:
\begin{align}
    \tau_j(t)= 
    \begin{cases}
    \tau_{j}(t-1)\qquad & \text{ if } j\neq J_t
   \\f(\tau_{j}(t-1),\tau_{j}(0))& \text{ if } j= J_t
    \end{cases}, j\in [n]
    \label{eq:pop}
\end{align}
where $f:\mathbb{R}\times \mathbb{R}\to \mathbb{R}$ is a strictly increasing function which captures the increase in a group's popularity due to the addition of a new member and its subsequent effect on the group's attractiveness to new members. For tractability, we assume that $f(\cdot,\cdot)$ is the same for all groups and fixed over time. If $f(x,y), x,y \in \mathbb{R}$ is a linear function of $x$ for any fixed $y$, then we recover the PA model in \cite{simon1957models,barabasi1999emergence}. On the other hand, if $f(x,y)$ is concave in $x$ for any fixed  $y$, then an increase in the popularity of an unpopular group increases its attractiveness to new users more significantly than a similar increase in the popularity of an already popular group. On the other hand, a convex $f(\cdot)$ creates the opposite effect.


\begin{Remark}
 We have assumed that at each step, there exists a user which is not already a member of $r_{J_t}$. We will show that due to the sparsity conditions considered in this work, the probability that there exists a group for which every user is its member, vanishes exponentially in the number of users as the graph becomes larger (Proposition \ref{prop:2}). However, for completeness, 
we assume that if every user is already a member of $r_{J_t}$ (i.e. if $\mathcal{U}_{J_t}(t-1)=[m]$), then an edge is not added in this step, the group popularities are updated as usual, and the generation process advances to the next step.
\end{Remark}
\begin{Remark}
We study bipartite graphs where the edges are binary-valued, i.e. a single edge between a given user and a given group is either present or absent. A natural extension is to consider edges with non-binary attributes and multigraphs. The attribute captures the nature of a users' group membership, e.g. group administrator, active member, etc. Inclusion of such information in the network graph may assist the attacker in deanonymizing the victim.
The information theoretic derivations provided in the next sections can be extended in a straightforward manner to graphs with attributed edges and multigraphs, where attributes are taken from an arbitrary finite set, and a finite number of edges is allowed between each two vertices, respectively. 
\end{Remark}



In this work, we focus on the particular choice of $f(x,y)=$ $ ((x-y)^{\frac{1}{\alpha}}+1)^\alpha+y, \alpha\in (0,1]$. This choice recovers several models for bipartite networks studied in prior works ---- such as equiprobable edges model, SB model, and linear and sublinear PA model --- by taking different values of $\alpha$ as described next. The parameter $\alpha$ is an intrinsic network parameter. In this case, Equation \eqref{eq:pop} can be rewritten as:
\begin{align*}
    \tau_j(t)=
    \begin{cases}
     \tau_j(t-1)\qquad &\text{ if } j\neq J_t
    \\D^\alpha_{t-1,j}+\tau_j(0) &\text{ if } j=J_t, D_{t-1,j}<n\\
   \tau^{\alpha}_j(t-1)+\tau_j(0) & \text{otherwise}
           \end{cases},
\end{align*}
where             $j\in [n]$, and $t\in [\Delta]$.  At a high level, $\alpha$ determines the effect of  the groups' sizes on the membership choices of new users, where larger $\alpha$ means that  the group-size plays a significant role in attracting new users, with large groups being more attractive, and at the other end of the spectrum, if $\alpha \to 0$, then the group popularities are constant through the generation process regardless of the group sizes. We focus on $\alpha\leq 1$ which leads to linear or sublinear PA and has been shown to be a suitable model for various networks of interest \cite{capocci2006preferential,newman2001clustering,albert2005scale,borrel2005preferential,van2009random,bloznelis2015random}.

\begin{Definition}[\textbf{Ground-truth Parameters}]
\label{Def:Models}
The ground-truth statistics are parametrized by $(n,m,\alpha, \Delta, $ $(\tau_j(0))_{j\in [n]})$. The following scenarios are considered in this work:
\\\textbf{$\alpha$-Preferential Attachment ($\alpha$-PA):} This is a generalization of the PA model, where $f(x)=((x-y)^{\frac{1}{\alpha}}+1)^{\alpha}+y, \alpha\in (0,1]$ and initial popularities are $\tau_j(0)=\tau_{j'}(0)=1, j,j'\in [n]$.
\\\textbf{Stochastic Blocks (SB):}  We take $\alpha \to 0$ and $\tau_j(0)\in \mathcal{T}$, where $ \mathcal{T}$ is a finite set. The collection of subsets $\mathcal{C}_\tau= \{r_j: \tau_j(0)= \tau\}, \tau\in \mathcal{T}$ are called the communities of social network groups.
\end{Definition}

\begin{Remark}
As a special case of the SB model, let us take $\alpha \to 0$ and $\tau_j(0)=\tau_{j'}(0), j,j'\in [n]$. Then, $f(x,y)=x$ for all $x,y\in \mathbb{R}$, and $\tau_j(t)=\tau_{j'}(t), j,j' \in [n], t\in [\Delta]$.
We call this the Independent and Equiprobable Edges (IEE) scenario. This is analogous to the Erd\"os-R\'enyi model for non-bipartite graphs \cite{erodos1959random}, and was studied in \cite{shirani2018optimal}. In this case, $P_j(t)= P_{j'}(t), j\in [n], t\in [\Delta]$, and  the groups are equally likely to attract new users regardless of their current number of members.
\end{Remark}
 
\begin{Remark}
In the SB scenario, we have $P_j(t)= \frac{\tau}{\sum_{\tau'\in \mathcal{T}}
\tau'|\mathcal{C}_{\tau'}|}, r_{j}\in \mathcal{C}_\tau, t\in [\Delta], \tau \in \mathcal{T}$. So,  the groups which belong to the same community $\mathcal{C}_\tau, \tau\in \mathcal{T}$ are equally likely to attract new users regardless of their current number of members. Groups may be classified into different communities based on the shared interests of their users, e.g. age group, profession, etc.
This model resembles the stochastic block model for social network friendship graphs \cite{florescu2016spectral,yen2020community}. 
\end{Remark}

\begin{Remark}
In the $\alpha$-PA scenario,
if $\alpha=1$ and $\tau_j(0)=\tau_{j'}(0)=1, j,j'\in [n]$, we have $f(x,y)=x+1, x\in \mathbb{R}$ and the model becomes the well-studied (linear) PA model.
In this case, $P_j(t)= \frac{D_{t-1,j}+1}{t+n-1}, j\in [n], t\in [\Delta]$, and the group sizes follow a power-law. This is in agreement with empirical studies of social network group memberships (e.g. \cite{capocci2006preferential,newman2001clustering}), where such power-law behavior has been observed.
\end{Remark}

\begin{Remark}
 In practice, the ground-truth statistics parametrized by $(n,m, \Delta, (\tau_j(0))_{j\in [n]}, \alpha)$ are not available to the attacker. Rather, the attacker acquires an estimate of these parameters as in \cite{kunegis2013preferential} based on prior observations of the bipartite network.  \end{Remark}

\subsection{The Scanned Graph}
\label{sec:scanned_graph}
As described in previous sections, the first phase of the fingerprinting attack is the passive phase, in which the attacker scans the social network 
for publicly available information regarding the users' group memberships. The attacker's observation of the ground-truth, acquired through this scanning process, is represented by the bipartite graph $\mathcal{G}_s= (\mathcal{U},\mathcal{R},\mathcal{E}_s)$, which is a partial and noisy observation of the users' group memberships. One reason for the noise in the scanned graph is that  some users may have made a subset of their group memberships hidden which results in edge omissions in the scanned graph. 
We model the resulting noise stochastically by assuming 
that the set of edges $\mathcal{E}_s$ in the scanned graph is generated randomly, conditioned on the set of edges $\mathcal{E}_0$ in the ground-truth graph. As discussed above, the difference between $\mathcal{E}_0$ and $\mathcal{E}_s$ is due to the privacy preferences of a specific user. As a result, we assume that the noise statistics in scanning a specific user-group edge $(u_k, r_j)$ is dependent on the corresponding user preference which is captured by the parameter $\gamma(k)\in \Gamma$, where $\Gamma$ is a finite set. This is formalized below.

\begin{Definition}[\textbf{Scanned Graph Statistics}]
 Let $P^{\gamma(k)}_{E_s|E_0}(\cdot|\cdot), \gamma(k) \in \Gamma, k\in [m]$ be a collection of conditional probability distributions, where $E_s$ and $E_0$ take binary values, and $\Gamma$ is a finite set. Let $R_{k,j}\triangleq \mathbbm{1}((u_k, r_j)\in \mathcal{E}_0)$ and $F_{k,j}\triangleq  \mathbbm{1}((u_k, r_j)\in \mathcal{E}_s), k\in [m], j\in [n]$. Then, 
 \begin{align*}
     P(\mathcal{E}_s|\mathcal{E}_0)=
     \prod_{k\in [m], j\in [n]} P^{\gamma(k)}_{E_S|E_0}(  F_{k,j}|R_{k,j}).
 \end{align*}
 In particular, the following Markov chains are assumed:
 \begin{align*}
     F_{k,j}\leftrightarrow R_{k,j},k \leftrightarrow (F_{k',j'},R_{k',j'})_{(k',j')\neq (k,j)}, k\in [m], j\in [n].
 \end{align*}
\end{Definition}

\begin{Example}[\textbf{Erasure Model for $\mathcal{G}_s$}]
Assume that the attacker scans a social network to acquire the scanned graph. The attacker observes a subset of the true group memberships of users \cite{kruegel} since some users choose to keep their membership in certain groups private. As a result, the scanned graph $\mathcal{G}_s$ consists of a sampled subset of the edges in the ground-truth $\mathcal{G}_0$. For simplicity, let us assume that the membership of user $u_k$ in group $r_j$ is publicly available  with probability $1-s_{k}, k\in [m], j\in [n]$, where $s_{k}\in [0,1]$. Then,  
\begin{align*}
&Pr(\mathcal{E}_s|\mathcal{E}_0)=\mathbbm{1}(\mathcal{E}_s\subset {\mathcal{E}_0})
\times
\\&
\prod_{k\in [m]}s_k^{|\mathcal{R}'_k|}(1-s_k)^{|\mathcal{R}_k|-|\mathcal{R}'_k|},
\end{align*}
where $\mathcal{R}_k$ and $\mathcal{R}'_{k}$ are the groups in which $u_k, k\in [n]$ is a member of in $\mathcal{G}_0$ and $\mathcal{G}_s$, respectively.
\end{Example}

\begin{Remark}
We assume that the attacker does not have knowledge of the users' privacy preferences, i.e it does not know the value of $\gamma(k), k\in [m]$ in $\Gamma$. The attacker only has access to the statistics $P^{\gamma}_{E_s|E_0}, \gamma\in \Gamma$. 
\end{Remark}


\subsection{Query Responses}
\label{sec:query_responses}
In the active phase of the attack, the attacker targets a victim, and actively queries its group memberships.  For instance, the victim visits a malicious website, and the attacker uses browser history sniffing techniques to query the victim's group memberships. The attacker may query the victim's group memberships sequentially by sending a single query regarding the victim's membership in a group at each step of the active attack, receiving a response, and deciding on the next query \cite{wondracek2010practical}. Alternatively, it may query a batch of group memberships simultaneously \cite{gulmezoglu2017perfweb, smith2018browser,shusterman2019robust}.
In this work, we focus on the first scenario, where the queries are made sequentially, one after the other. However, the analysis can be extended to the second scenario, where queries are made in batches, in a straightforward manner. 

The objective is to deanonymize the victim based on their group membership fingerprint.
We model the victim stochastically by assuming that it is chosen randomly from the user set. In general, the users are not equally likely to be a victim of an attack, For instance, users are not equally likely to visit a malicious website, risk-averse users are less likely to be the victim of a fingerprinting attack compared to risk-taker users. 
As a result, we assume that the victim $u_M$ is chosen from $\mathcal{U}$ based on an underlying distribution $P_M$.

\begin{Remark}
 In this work, following the conventional approach in privacy and security literature, we investigate a  `genie-aided' attacker by assuming access to $P_M$ in order to derive theoretical guarantees for users' privacy. However, it should be noted that, in practice, the attacker may only have an estimate $\hat{P}_M$ of $P_M$ or it may not have any prior knowledge of these statistics at all. In such cases, the attack strategies investigated in the following sections may be extended naturally, and their probability of success can evaluated with respect to a `worst-case' distribution $\hat{P}_M$.
\end{Remark}

Let us assume that the attacker queries the group memberships of the victim $u_M$ in the sequence of groups $(r_{j_1},r_{j_2},\cdots,r_{j_\ell}), j\in [n]$ in $\ell\in \mathbb{N}$ queries, and receives the binary vector of query responses $Y_1,Y_2,\cdots,Y_\ell$, where $Y_i=1$ indicates a positive response and $Y_i=0$ a negative response.
Generally, query responses are noisy since browser history sniffing techniques are imperfect and only provide noisy observations of the victim's browsing history. That is, $Y^\ell$ is a noisy version of the true group membership indicators $(R_{j_1},R_{j_2},\cdots, R_{j_\ell})$.
The noise statistics are determined by the users' software (e.g. browser \cite{smith2018browser}) and hardware specifications (e.g. CPU and memory specifications  \cite{gulmezoglu2017perfweb}) ,  and depend on the type of history sniffing attack.  However, these statistics do not depend on the specific website or group whose membership is being queried. This dependency is captured by the parameter $\theta(M)$, where $\theta:[m]\to \Theta$, and $\Theta$ is a finite set. The following definition formalizes the stochastic model for the query responses.  

\begin{Definition}[\textbf{Noisy Query Responses}]
\label{def:QR}
Let $\ell\in \mathbb{N}$ and let $P^{\theta}_{Y|R}, \theta\in \Theta$ be a collection of probability distributions, where $Y$ and $R$ are binary variables and $\Theta$ is a finite set. For the sequence $j_1,j_2,\cdots,j_{\ell}\in [n]$, assume that victim's fingerprint is $(R_{j_1},R_{j_2},\cdots, R_{j_\ell})$ and the received query responses are $Y_1,Y_2,\cdots,Y_\ell$. Then, 
\begin{align*}
P(Y^\ell=y^\ell|   (R_{j_i})_{i\in [\ell]}=r^\ell) =\prod_{i=1}^\ell P^{\theta(M)}_{Y|R}(y_i|r_i), y^{\ell}, r^{\ell}\in \{0,1\}^{\ell},
\end{align*}
where the parameter $\theta(M)$ takes values from $\Theta$ and its value depends on the victim's index $M$.
\end{Definition}
\begin{Remark}
In practice, the attacker does not have access to the statistics $P^{\theta(k)}_{Y|R}(y_i|r_i), k\in [m]$. Rather, it may query the victim's software and hardware specifications to acquire $\theta(k)$, and then estimate the noise statistics
based on prior observations of the querying process with these specifications and based on the history sniffing technique used by the attacker. This is in contrast with the noise model in the scanned graph $P^{\gamma(k)}_{Y|R}(y_i|r_i), k\in [m]$, where the attacker has no means of learning the user's privacy preferences $\gamma(k)$. 
\end{Remark}

To summarize, an active bipartite network deanonymization setup is characterized as follows.

\begin{Definition}[\textbf{Active Bipartite Network Deanonymization}]
 An active bipartite network deanonymization setup is characterized by parameters $(n,m, \Delta, \Theta, \Gamma,  \alpha,$ $ (\tau_j(0))_{j\in [n]}, P_M, (P^{\gamma(k)}_{E_S|E_0})_{k\in [m], \gamma \in \Gamma}, (P^{\theta(k)}_{Y|R})_{k\in [m], \theta\in \Theta})$, where $n$ is the number of groups, $m$ the number of users, $P_M$ determines the victim's ($u_M$) distribution among the users $\mathcal{U}$,  $P^{\theta(k)}_{Y|R}, \theta(k) \in \Theta$ is the query response noise statistics for user $k\in [m]$, $P^{\gamma(k)}_{E_S|E_0}, \gamma(k)\in \Gamma$ is the scanned graph noise statistics for user $k\in [m]$, $\alpha$ is the network growth parameter, 
 $\Delta$ is the total number of edges, and $(\tau_j(0))_{j\in [n]}$ are the initial group popularities. 
\end{Definition}

\subsection{Attack Strategy}
\label{sec:strategy}
Given the scanned graph $\mathcal{G}_s$ acquired by scanning the ground-truth $\mathcal{G}_0$,
the attacker's objective is to identify the victim using the minimum number of queries possible, and with small probability of error. An attack strategy determines the sequence of queries made by the attacker, and identifies the victim based on the query responses.
It
consists of a sequence of query functions $x_t(\cdot,\cdot), t\in \mathbb{N}$ and identification functions $Id_t(\cdot,\cdot), t\in \mathbb{N}$, where at time\footnote{Note that we have used the variable `t' to refer to two different time quantities. One is the steps in the ground-truth generation process ($t\in [\Delta])$ in Section \ref{sec:ground_truth}, and the other one is the number of queries sent in the active phase of the attack ($t\in \mathbb{N}$) which is discussed here.} $t$, the query function $x_t(\mathcal{G}_s,Y^{t-1})$ takes the scanned graph $\mathcal{G}_s$ and the received query responses $Y^{t-1}$ as input, and outputs the group $r_{j_t}= x_t(\mathcal{G}_s,Y^{t-1}), j_t\in [n]$ whose connection with the victim is to be queried next. Assume that the response $Y_t$ is received. The identification function $Id_t(\mathcal{G}_s,Y^t)$ compares the received query responses $Y^t$ with the users' fingerprints in the scanned graph $\mathcal{G}_s$, and either outputs the identity of the victim, or indicates that the identity cannot be determined yet, hence the attack continues with the next query.
This is formalized below. 

\begin{Definition}[\textbf{Attack Strategy}]
 Consider an ABND scenario parametrized by $(n,m, \Delta, \Theta, \Gamma,  \alpha,$ $ (\tau_j(0))_{j\in [n]}, $ $P_M, (P^{\gamma(k)}_{E_S|E_0})_{k\in [m], \gamma \in \Gamma}, (P^{\theta(k)}_{Y|R})_{k\in [m], \theta\in \Theta})$. An attack strategy consists of a sequence of query functions $x_t: \{0,1\}^{m\times n}\times\{0,1\}^{(t-1)}\to\mathcal{R}, t\in \mathbb{N}$ and  identification functions $Id_t: \{0,1\}^{m\times n} \times \{0,1\}^t \to \mathcal{U} \cup \{e\}$, where $x_t(\mathcal{G}_s, Y^{t-1})$ outputs the group whose edge connection with the victim is queried at time $t$, and $Id_t(\mathcal{G}_s,Y^t)$  either outputs the victim's identity among the user set $\mathcal{U}$ or outputs `$e$' in which case further queries are made and the attack continues. 
Let $Q= min\{t\in \mathbb{N}:Id_t(\mathcal{G}_s,Y^t)\in \mathcal{U}\}$. Then, the
probability of error $P_e$ and  expected number of queries $\overline{Q}$ are defined as:
\begin{align*}
    &P_e((x_t,Id_t)_{t\in \mathbb{N}})\triangleq 
    P(Id_Q( \mathcal{G}_s,Y^Q)\neq u_M)\\
    &\overline{Q}((x_t,Id_t)_{t\in \mathbb{N}})\triangleq 
    \mathbb{E}(  Q), 
\end{align*}
 where
 the probabilities are with respect to $M, \mathcal{G}_0, \mathcal{G}_s$ and $Y_t, t\in [Q]$. 
\end{Definition}

 \begin{Definition}[\textbf{Minimum Expected Queries}] 
 For the ABND problem characterized by $(n,m, \Delta, \Theta, \Gamma,  \alpha,$ $ (\tau_j(0))_{j\in [n]}, P_M, (P^{\gamma(k)}_{E_S|E_0})_{k\in [m], \gamma \in \Gamma}, (P^{\theta(k)}_{Y|R})_{k\in [m], \theta\in \Theta})$, and error probability $\epsilon>0$,  the minimum expected number of queries is defined as:
 \[{Q}^*_{\epsilon}\triangleq\inf_{(x_t,Id_t)_{t\in \mathbb{N}}}\{\overline{Q}((x_t,Id_t)_{t\in \mathbb{N}})| P_e((x_t,Id_t)_{t\in \mathbb{N}})\leq \epsilon\}.\]
\end{Definition}

Our objective is to investigate the necessary and sufficient conditions under which an attacker can deanonymize the victim reliably (i.e. with vanishing error probability) over asymptotically large bipartite networks. That is, we want to investigate the problem when the number of users $m$ grow asymptotically large. In particular,  based on observations of real-world social networks (e.g. \cite{kunegis2013preferential,zheleva2009co}), we investigate the ABND problem under the following asymptotic regime:
\begin{itemize}
    \item \textbf{Number of Groups:} The number of groups $n$ grows linearly in $m$, i.e. $m=\beta n$ for a fixed $\beta>0$. 
    \item \textbf{Noise Parameters:} The sets $\Theta, \Gamma$ and $P^{\theta}_{Y|R}, P^{\gamma}_{E_s|E_0}, \theta\in \Theta, \gamma\in \Gamma$ are fixed in $m$. This is justified since $P^{\gamma}_{E_s|E_0}$ and $P^{\theta}_{Y|R}$  are determined by the users' privacy preference options in the social network, and their software/hardware specifications, respectively, and do not change as the number of users increases asymptotically.
    \item  \textbf{Sparsity:} The average number of  groups in which any given user is a member of is constant as the network grows. That is, $\Delta=  \mu n= \frac{\mu}{\beta}m, \mu \geq 1$, so that the average group size $\mu$ is constant in $n$.
\item  \textbf{Victim's Distribution:} The users' likelihood of being the victim decreases inversely in $m$, that is $P_M(u_k)= \frac{c_k}{m}$, where $\sum_{k\in [m]}{c_k}=m$ and $c_k<\lambda, k\in [m]$ as $n\to \infty$ for some constant $\lambda>0$. 
\end{itemize}

\section{Memory Structure of the ground-truth Edges}
\label{sec:mem}
The scanned graph and the query responses provide the attacker with two noisy observations of the victim's group  membership fingerprint in the ground-truth. The attacker identifies the victim by reconciling the query responses with the user fingerprints in the scanned graph and finding a unique match (e.g. jointly typical pair of fingerprint and query response vectors). 
 One major obstacle in analyzing the fundamental performance limits of attack strategies is the \emph{memory structure} in the user's fingerprint induced due to the generation model of the ground-truth described in Section \ref{sec:ground_truth}.  That is, the generation model induces correlation among the users' membership in different groups. This prohibits the conventional methods such as type analysis and large deviations techniques  which have been used in deriving theoretical performance limits in similar scenarios in group testing \cite{naghshvar2013active} and communications \cite{burnashev1976data} problems, as well as the analysis techniques in prior work on ABND \cite{shirani2017information,shirani2018optimal}. In this section, we show that under the sparsity assumption on the total number of edges that $\Delta=\mu n$, the memory in the users' fingerprint is weak, and its joint distribution is well-approximated by a product distribution.  The derivations are used in the next sections, where we propose attack strategies and derive sufficient conditions for their success. These are also of independent interest in analyzing degree distributions of vertices in bipartite networks. 
 
 \subsection{Weakly Correlated Group Sizes}
 Let us recall that the size of group $r_j, j\in [n]$ at step $t\in [\Delta]$ of the generation process is defined as $D_{t,j}= |\mathcal{U}_j|, j\in [n]$. 
 As a first step towards investigating the correlation among users' memberships in different groups, we study the joint moments of $(D_{\Delta,j})_{j\in [n]}$ and show that they converge to a finite constant as $n\to \infty$, and $\Delta=\mu n\to \infty$.
 \begin{Proposition}[\textbf{Group Size Correlation}]
 \label{prop:1}
Let $0<\alpha<1$.  For a ground-truth graph generated according to the $\alpha$-PA model, the following holds:
\begin{align}
& \label{eq:prop1:1}\mathbb{E}(D_{\Delta})=\mu,
\\\label{eq:prop1:2}
    & \mathbb{E}(D^2_{\Delta,j})= O(1), j\in [n],
     \\&\label{eq:prop1:3} \mathbb{E}(D_{\Delta,i}D_{\Delta,j})= \mu^2+O(\frac{1}{n}), i\neq j,
          \\&\label{eq:prop1:2.5} \mathbb{E}(D_{\Delta,1}D_{\Delta,2}\cdots D_{\Delta,\zeta})= \mu^\zeta(1+\zeta O(\frac{1}{n})), \zeta\in [n],
    \\\label{eq:prop1:4}
   & \mathbb{E}(D_{\Delta,1}^2D_{\Delta,2}D_{\Delta,3}\cdots D_{\Delta, \zeta})\leq
    \mu^{\zeta-1} \mathbb{E}(D^2_{1,\Delta})
   , \zeta\in [n],
    \\&\label{eq:prop1:5}
    \mathbb{E}(D_{\Delta,1}D_{\Delta,2}D_{\Delta,3}\cdots D_{\Delta, \zeta})\leq \mu^{\zeta}, \zeta\in [n].
\end{align}
 \end{Proposition}
 \begin{proof}
Appendix \ref{app:prop:1}.
\end{proof}
\subsection{Almost Memoryless Fingerprints} 
Next, we prove that under the sparsity condition $\Delta=\mu n$, the fingerprints in the ground-truth are `almost' memoryless. Let the number of groups in which a user is a member be denoted by $C_{i}\triangleq|\mathcal{R}_i|, i\in [m]$. The users' memberships in different groups are correlated due to the ground-truth generation model. We are interested in investigating this correlation. As a first step, we show in the following that each user's fingerprint is sparse (i.e. has few ones). 

\begin{Proposition}[\textbf{Sparsity of the User Fingerprint Vector}]
\label{prop:2}
 Let $\alpha\in (0,1]$, $\mu \in \mathbb{N}$, and  $\beta>0$. For a ground-truth graph generated according to the $\alpha$-PA model with $n\in \mathbb{N}$ groups, $m=\beta n$ users, and $\Delta=\mu n$ edges, there exists a constant $c>0$ such that:
\begin{align}
   & P(C_i\geq \ell)\leq c2^{-nD_b(\frac{\mu}{m}(1+\psi)||\frac{\mu}{m})},
\end{align}
where $\ell=\frac{1}{\beta}\mu(1+\psi)$, $\psi \in (0, \frac{m}{\mu}-1)$, and $D_b(p||q)=p\log{\frac{p}{q}}+ (1-p)\log{\frac{1-p}{1-q}}$ is the binary  Kullback-Leibler divergence. In particular, let $\psi_n>0, n\in \mathbb{N}$ such that  $\psi_n=\omega(1)$. Then, 
\begin{align}
\label{eq:prop2}
       & P(C_i\geq \psi_n)\to 0,\text{ as }n\to\infty.
\end{align}
\end{Proposition}
\begin{proof}
Appendix \ref{app:prop:2}.
\end{proof}
The next proposition shows that the distribution of the fingerprint of each user in the ground-truth graph is close to a  memoryless distribution.
\begin{Proposition}[\textbf{Memoryless Fingerprints in $\alpha$-PA}]
\label{prop:3}
 Let $\alpha\in (0,1]$. For a ground-truth graph generated according to the $\alpha$-PA model, consider the partial fingerprint $\mathbf{R}\triangleq (R_{i,j_k})_{k\in [n']}, j_k\in [n], n' \in [n]$ of user $u_i, i\in [m]$. The following holds:
 \begin{align*}
&(1-\frac{n'\mu}{m}) \prod_{k=1}^{n'}P_{R}(s_k)
\leq 
     P_{\mathbf{R}}(s^{n'})\leq  
 e^{\frac{\mu}{\beta}}\prod_{k=1}^{n'}P_{R}(s_k), s^{n'}\in \{0,1\}^{n'},
 \end{align*}
 where $P_R(\cdot)= P_{R_{i,j}}(\cdot), i\in [m], j\in [n]$. Furthermore, assume that $n'>\frac{m}{\mu}$ and $\sum_{i=1}^{n'}\mathbbm{1}(s_i=1)=o(n)$ for some constant finite number $C>0$.  Then, there exists $c'>0$ whose value only depends on $\mu$ and $\beta$ such that:
 \begin{align*}
   & c'\prod_{k=1}^{n'}P_{R}(s_k)(1+o(1))\leq P_{\mathbf{R}}(s^{n'})
  \leq \prod_{k=1}^{n'}P_{R}(s_k)(1+o(1)), s^{n'}\in \{0,1\}^{n'},
 \end{align*}
 as $n \to \infty$. 
\end{Proposition}

\begin{proof}
Appendix \ref{app:prop:3}.
\end{proof}

In the SB scenario,  edge probabilities do not change during the generation process  and the number of groups associated with each user follows a (truncated) Binomial distribution with parameters $(\Delta, \frac{\mu}{\Delta})$. As a result, it is straightforward to establish the memoryless property of the fingerprints using standard arguments based on law of large numbers. It should be noted that there is correlation among group sizes in this case since for instance 
\begin{align*}
\mathbb{E}(D_{1,\Delta}D_{2,\Delta})&= \mathbb{E}(D_{1,\Delta}\mathbb{E}(D_{2,\Delta}|D_{1,\Delta}))
= \mathbb{E}(D_{1,\Delta})\mathbb{E}(D_{2,\Delta})(1-\frac{\mathbb{E}(D_{1,\Delta})}{\Delta}),\end{align*}
where we have used the smoothing property of expectation. However, the correlation in the user fingerprint vectors is weak and it can be observed that for any binary vector $s^n\in \{0,1\}^n$, we have \[(1-\frac{|w_H(s^n)|}{\Delta})^{w_H(s^n)}\leq \frac{\prod_{k=1}^{n}P_{R}(s_k)}{ P_{(R_{i,j_k})_{k\in [n]}}(s^{n'})}\leq  (1+\frac{|w_H(s^n)|}{\Delta})^{w_H(s^n)},\]
where $w_H(\cdot)$ is the Hamming weight. Note that $w_H((R_{i,j_k})_{k\in [n]})\to \mu$ with probability one due to concentration of measure. So, we conclude that $\frac{\prod_{k=1}^{n}P_{R}(s_k)}{ P_{(R_{i,j_k})_{k\in [n]}}(s^{n'})}\approx 1$. The following proposition formalizes this statement. The proof is straightforward and is omitted for brevity. 

\begin{Proposition}[\textbf{Memoryless Fingerprints in SB}]
\label{prop:4}
 For a ground-truth graph generated according to the SB model, consider the partial fingerprint $\mathbf{R}\triangleq (R_{i,j_k})_{k\in [n']}, j_k\in [n], n' \in [n]$ of user $u_i, i\in [m]$. The following holds:
 \begin{align*}
P_{\mathbf{R}}(s^{n'})= o(1),  s^{n'}\in \{0,1\}^n:w_{H}(s^{n'})> \mu(1+\omega(1)),
 \end{align*}
 Furthermore, 
 \begin{align*}
  P_{\mathbf{R}}(s^{n'})=
  (1+o(\frac{1}{n}))\prod_{k=1}^{n'}P_{R}(s_k),
 \end{align*}
as $n \to \infty$, where $s^{n'}\in \{0,1\}^n:w_{H}(s^{n'})= \mu(1+O(1))$. 
\end{Proposition}

\section{Sufficient Conditions for Successful Deanonymization}
\label{sec:ITS}
In this section, we derive sufficient conditions on the network parameters and the expected number of queries under which the attacker can successfully  deanonymize the victim with vanishing probability of error as $m\to \infty$.  Initially, we make simplifying assumptions on the scanning and querying noise statistics and develop the tools to study the more complex formulation in the next steps. We relax these assumptions in steps and derive general  theoretical guarantees for successful deanonymization.
\subsection{Identical Scanning Noise and Noiseless Query Responses}
As a first step, we consider the scenario in which the scanning noise is identical for all users, i.e. $\Gamma=\{1\}$, and the query responses are received noiselessly, i.e. $|\Theta|=1, P^{1}_{Y|E_0}(y|s)=\mathbbm{1}(y=s), y,s\in \{0,1\}$. 

Let us focus on the $\alpha$-PA model for a given $\alpha\in (0,1]$. We generalize the information threshold strategy (ITS), which was introduced in \cite{shirani2018optimal}, where we studied a scenario in which the ground truth is generated according to the IEE model. It was shown in \cite{shirani2018optimal} that 
the strategy is asymptotically optimal under IEE model  --- in terms of expected number of queries necessary for successful deanonymization with vanishing error. In the ITS, the attacker queries the group memberships of the victim starting from the first group $r_1$ and continuing by increasing the group index (i.e. $x_t=r_t, t\in [n]$), until a particular stopping criterion is met. To explain the stopping criterion, let us define the \textit{information value} $I_{k}(t), k\in [m], t\in [n]$ of user $u_k$ and time $t$ as follows:
\begin{align*}
    &I_0(k)= \log{P_M(k)}, k\in [m],\\
    &I_t(k)=
\sum_{i=1}^t    \log{\frac{P_{E_0|E_s}(y_i|f_{k,i})}{P_{E_0}({y_i})}}+I_0(k), k\in [m], t\in [n]
\end{align*}
where $(f_{k,i})_{i\in [t]}\in \{0,1\}^t$ is the realization of the partial fingerprint of user $u_k$ in the scanned graph (i.e. $(F_{k,i})_{i\in [t]}=(f_{k,i})_{i\in [t]}$), the vector $y^t\in \{0,1\}^t$ is the realization of the vector of query responses (i.e. $Y^t=y^t$), and
\begin{align}
    P_{E_0|E_s}(y|f)\triangleq  \frac{P_{E_0}(y)P^1_{E_s|E_0}(f|y)}{\sum_{y'\in \{0,1\} }  P_{E_0}(y')P^1_{E_s|E_0}(f|y') }.
    \label{eq:distA}
\end{align}

The identification function $Id_t$ first determines whether the maximum information value of all users exceeds $\log{\frac{1}{\epsilon}}$, where the parameter $\epsilon>0$ affects the resulting probability of error. If there exists a user whose information value exceeds  $\log{\frac{1}{\epsilon}}$, that user is identified as the victim. Otherwise, the next query is made. So,
\begin{align}
    &x_t(\mathcal{G}_s,Y^t)=r_t, t\in [n]\\
&    Id_t(\mathcal{G}_s,Y^t)=  
    \begin{cases}
    u_k\qquad& \text{ if } \exists ! k\in [m]: I_t(k)>log{\frac{1}{\epsilon}}\\
    e& \text{Otherwise}
    \end{cases}, t\in [n] \label{eq:id}
\end{align}
We call this attack strategy the ITS due to the use of information thresholds for deanonymization. 

\begin{Theorem}
Consider the ITS described above with parameter $\epsilon>0$. Let $\overline{Q}_{ITS}$ be the resulting expected number of queries and $P_{e,ITS}$ the resulting probability of error. Then, in the $\alpha$-PA scenario with $\alpha\in (0,1]$:
    \begin{align} 
        &\overline{Q}_{ITS}\leq \frac{H(M)+\log{\frac{1}{\epsilon}}+ i_{\max}}{c'I(E_0;E_s)},  \label{eq:th:1}\\
        &P_{e,ITS}\leq \frac{\epsilon}{c'},
    \end{align}
    where $c'$ is from Proposition \ref{prop:3}, the mutual information is evaluated with respect to $P_{E_0,E_s}= P_{E_0}P_{E_s|E_0}$,  the distribution $P_{E_s|E_0}$ is given in \eqref{eq:distA}, the variable $E_0$ is Bernoulli with $P_{E_0}(1)= 1-P_{E_0}(0)= \frac{\mu}{m}$, and $i_{max}\triangleq \max_{y,f\in \{0,1\}}\log{\frac{P_{E_0|E_s}(y|f)}{P_{E_0}(y)}}$.
\label{th:1}
\end{Theorem}
\begin{proof}
Appendix \ref{app:th:1}.
\end{proof}

\begin{Remark}
The coefficient $c'$ in the denominator of $\frac{H(M)+\log{\frac{1}{\epsilon}}+ i_{\max}}{c'I(E_0;E_s)}$ can be improved in special cases based on the value of $\alpha$. For instance, it is shown in \cite{shirani2018optimal} that for the IEE model, where $\alpha\to 0$, the denominator $c'I(E_0;E_s)$ can be replaced by $I(E_0;E_s)$ to derive an asymptotically optimal bound.
\end{Remark}

Next, we focus on the SB model. Let $P^{\tau}_{E_0}(1)=\frac{\tau}{\sum_{\tau'\in \mathcal{T}}\tau'|\mathcal{C}_{\tau'}|}, \tau\in \mathcal{T}$, and let us assume without loss of generality that $P^1_{E_0}(1)\leq P^2_{E_0}(1)\leq \cdots \leq P^{|\mathcal{T}|}_{E_0}\leq \frac{1}{2}$. Then, the ITS  query function queries the groups starting with most popular communities of groups. To elaborate, assume that   $\mathcal{T}= \{1,2,\cdots,|\mathcal{T}|\}$ and  $\tau_0(j)\geq \tau_0(j'), j>j'$. Then $x(\mathcal{G}_s,Y^{t-1})=r_t, t\in [n]$. 
Note that we have assumed that the attacker knows the community membership of the groups. In the absence of this information, the attacker may potentially  extract the group's community memberships using $\mathcal{G}_s$. 

The stopping criterion is modified as follows. The information value $I_{k}(t), k\in [m], t\in [n]$ of user $u_k$ and time $t$ is:
\begin{align*}
    &I_0(k)= \log{P_M(k)}, k\in [m],\\
    &I_t(k)=\sum_{\tau\leq \tau'}
\sum_{\ell=1}^{|\mathcal{C}_{\tau}|}    \log{\frac{P_{E_0|E_s}(y_\ell|f_{k,\ell})}{P_{E^{\tau}_0}({y_\ell})}}+
\\&\sum_{i=0}^{i'}
\log{\frac{P_{E_0|E_s}(y_i|f_{k,i})}{P_{E^{\tau'}_0}({y_i})}}+
I_0(k), k\in [m], t\in [n]
\end{align*}
where $t= \sum_{\tau \leq \tau'}|\mathcal{C}_{\tau}|+i', i' \leq |\mathcal{C}_{\tau'+1}|$.
\begin{Theorem}
\label{th:2}
 In the SB scenario, let $\mathcal{T}= \{1,2,\cdots,|\mathcal{T}|\}$ and assume that $\tau_0(j)\geq \tau_0(j'), j>j'$, then: 
        \begin{align*}
        &\overline{Q}_{ITS}\leq
        \sum_{\tau\leq \tau^*}|\mathcal{C}_{\tau}|+ i^*,
         \qquad P_{e,ITS}\leq \epsilon,
    \end{align*}
    where $(\tau^*,i^*)$ are defined as
    \begin{align*}
       &\tau^*\triangleq\min_{\tau\in \mathcal{T}}\bigg\{\tau:\psi \leq \sum_{\tau'\leq \tau+1}|\mathcal{C}_\tau|I_{\tau}(E_0;E_s)\bigg\}, \\
       &i^*\triangleq \min_{i\in [|\mathcal{C}_{\tau^*}|]}\bigg\{i: 
      \psi \leq \sum_{\tau\leq \tau^*}|\mathcal{C}_\tau|I_{\tau}(E_0;E_s)+iI_{\tau^*+1}(E_0;E_s) \bigg\},
      \\&\psi\triangleq H(M)+\log{\frac{1}{\epsilon}+i_{\max}},
    \end{align*}
    the mutual information $I_{\tau}(E_0;E_s)$ is evaluated with respect to $P^{\tau}_{E_0,E_s}= P^{\tau}_{E_0}P_{E_s|E_0}$, the variable $E_0$ is Bernoulli with parameter $P^{\tau}(E_0=1)= \frac{\tau}{\sum_{\tau'\in \mathcal{T}}\tau'|\mathcal{C}_{\tau'}|}\frac{\mu}{\beta}$, and $P_{E_s|E_0}$ is given in \eqref{eq:distA}. 
\end{Theorem}
The proof for the upper bound on $\overline{Q}_{ITS}$ follows similar arguments as Theorem \ref{th:1}, and uses \eqref{eq:th1:2} along with the fact that
 \begin{align*}
    \mathbb{E}(I_{T_{n'}(M)})
    =  \sum_{\tau\leq \tau'}|\mathcal{C}_\tau|I_{\tau}(E_0;E_s)+
    i'I_{\tau'+1}(E_0;E_s),
\end{align*}
where, $\mathbb{E}(T_{n'})=\sum_{\tau\leq \tau'} |\mathcal{C}_\tau|+i', i'\leq |\mathcal{C}|_{\tau'+1}$. 
 The derivation of the upper bound on the probability of error follows similar arguments as in the proof of Theorem \ref{th:1} along with Proposition \ref{prop:4}.

\subsection{Noisy Scan and Query Responses}
In this section, we extend the ITS to the case of noisy scan and query responses and arbitrary finite sets $\Gamma$ and $\Theta$. Note that as explained in Section \ref{sec:form}, the attacker does not have access to $\gamma(k), k\in [m]$ but has access to $\theta(M)$. To elaborate, the attacker knows the user device specifications, and hence it knows the query response noise statistics $P^{\theta(M)}_{Y|E_0}$, but does not know the users' privacy preferences, and hence it only knows that the scan knows statistics is given by one of the conditional distributions $P^{\gamma}_{E_s|E_0}, \gamma\in \Gamma$, where $\gamma$ may be different for different users. 

Let us focus on the $\alpha$-PA model for $\alpha\in (0,1]$. The query function is defined as in the previous section. The identification function is modified as follows. The attacker has access to $\theta(M)$ since it can query the victim's hardware and software specifications. So, it can find $P^{\theta(M)}_{Y|E_0}$ and use it in calculating the users' information values as in the previous scenario. As for the scan noise parameter, $\gamma$, the attacker computes $|\Gamma|$ different information values for each user, one for each value of $\gamma\in \Gamma$, and assigns the maximum resulting value as the information value of the user. This resembles the communication strategies used for communicating over compound channels when channel state information is unavailable \cite{boche2020}. So,

\begin{align*}
    &I_0(k)= \log{P_M(k)}, k\in [m],\\
    &I_t(k)=\max_{\gamma\in \Gamma}
\sum_{i=1}^t    \log{\frac{P^{\gamma,\theta(M)}_{Y|E_s}(y_i|f_{k,i})}{P^{\gamma}_{Y}({y_i})}}+I_0(k), k\in [m], t\in [n],
\end{align*}
where $P^{\gamma}_Y(\cdot)=\sum_{s\in \{0,1\}}P_{E_0}(s)P^{\gamma}_{Y|E_0}(\cdot|s)$, and $P^{\gamma,\theta(M)}_{Y|E_s}(y_i|f_{k,i})= \sum_{s\in \{0,1\}} P^{\theta(M)}_{E_0|E_s}(s|f_{k,i}) P^{\gamma}_{Y|E_0}(y|s), y,f_{k,i}\in \{0,1\}$, and:

\begin{align}
    P^{\Theta(M)}_{E_0|E_s}(s|f)\triangleq  \frac{P_{E_0}(s)P^{\Theta(M)}_{E_s|E_0}(f|s)}{\sum_{s'\in \{0,1\} }  P_{E_0}(s')P^{\Theta(M)}_{E_s|E_0}(f|s') }.
\end{align}

The following sufficient conditions for successful deanonymization are given in the following theorem. 

\begin{Theorem}
Consider the ITS described above with parameter $\epsilon>0$. Let $\overline{Q}_{ITS}$ be the resulting expected number of queries and $P_{e,ITS}$ the resulting probability of error. Then, in the $\alpha$-PA scenario with $\alpha\in (0,1]$:
    \begin{align*}
        &\overline{Q}_{ITS}\leq \sum_{\gamma\in \Gamma,\theta\in \Theta} P_{\Gamma,\Theta}(\gamma,\theta) \frac{H(M)+\log{\frac{1}{\epsilon}}+ i_{\max}}{c'I_{\gamma,\theta}(Y;E_s)},\\
        &P_{e,ITS}\leq \frac{|\Gamma|\epsilon}{c'},
    \end{align*}
    where $c'$ is from Proposition \ref{prop:3}, the mutual information is evaluated with respect to $P_{E_0,E_s}= P_{E_0}P_{E_s|E_0}$,  the distribution $P_{E_s|E_0}$ is given in \eqref{eq:distA}, the variable $E_0$ is Bernoulli with $P_{E_0}(1)= 1-P_{E_0}(0)= \frac{\mu}{m}$, $i_{max}\triangleq \max_{y,f\in \{0,1\}}\log{\frac{P_{E_0|E_s}(y|f)}{P_{E_0}(y)}}$, and $P_{\Gamma,\Theta}(\gamma,\theta)\triangleq \frac{|\{u_k| \theta(k)=\theta, \gamma(k)=\gamma\}|}{m}, \theta\in \Theta, \gamma\in \Gamma$.
\label{th:3}
\end{Theorem}
The proof follows similar argument to that of Theorem \ref{th:1}. The derivation of the bound on $\overline{Q}_{ITS}$ for a given choice of $\theta\in \Theta,\gamma \in \Gamma$ is unchanged since the number of queries needed to achieve the desired information value threshold does not increase with the modified information values, since users are assigned a higher information value by maximizing over $\gamma\in \Gamma$.
The bound on the probability of error follows the exact same steps as in the proof of Theorem \ref{th:1} with the additional step of using the union bound to bound the probability of error over the union of choices of $\Gamma$.

\begin{Remark}
Similar to the derivation of Theorem \ref{th:3} which extends Theorem \ref{th:1} to general scan and query noise statistics, Theorem \ref{th:2} can also be extended to derive sufficient conditions for the success of ITS under the SB model and general noise statistics. Again, the bound on the expected number of queries $\overline{Q}$ remains the same, but the upper-bound on the probability of error $P_e$ grows linearly in $|{\Gamma}|$. 
\end{Remark}
\begin{figure}[t]
 \centering \includegraphics[width=0.6\linewidth,height=2.3in, draft=false]{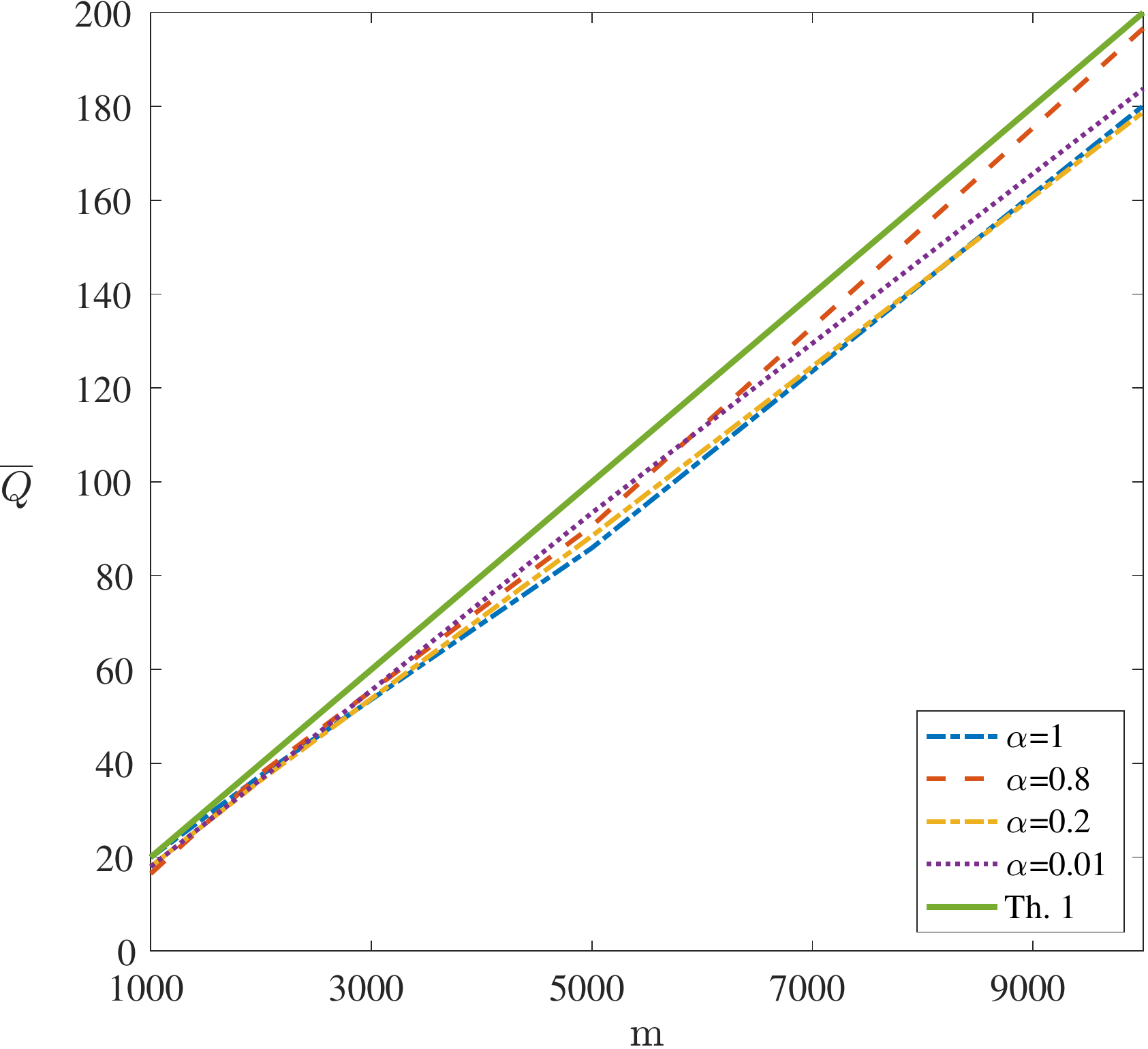}
  \caption{Expected number of queries necessary for success in the ITS under noiseless scanned and query response assumptions with a victim which is uniformly chosen among the users. The green (filled) line is the upper bound on the expected number of queries due to Theorem \ref{th:1}.}
 \label{fig:sim:1}
\end{figure}

\section{Simulation Results}
\label{sec:simul}
In this section, we provide several simulations of synthesized and real-world ABND attacks to verify the theoretical results presented in the previous sections and gain further intuition regarding the users' privacy risks under such attack scenarios. 

\subsection{Effect of Growth Parameter $\alpha$ on $\overline{Q}_{ITS}$}
As a first step, we consider a noiseless ABND scenario under the $\alpha$-PA generation model, where the scanned graph and the query responses are acquired noiselessly by the attacker, and the victim is equally likely to be any of the users. We wish to evaluate the effect of changing the preferential attachment parameter $\alpha\in (0,1]$ on the expected number of queries necessary for deanonymization under ITS. Note that in this case, ITS reduces to a simple strategy, where queries are made until the acquire responses have a unique match among the user fingerprints in the scanned graph since $P_{Y|E_s}(y|f)= \mathbbm{1}(y=f), y,f\in \{0,1\}$. Our analysis in Theorem \ref{th:1} predicts that $\overline{Q}_{ITS}$ grows linearly in $m\in \mathbb{N}$ since in the denominator in \eqref{eq:th:1} we have $I(E_0;E_S)= H(E_0)= \frac{m}{\mu}(\log{m}+o(\log{m}))$, and in the numerator we have $H(M)=\log{m}$ and $i_{max}=\log{m}$. In fact, Theorem \ref{th:1} predicts that $\overline{Q}\triangleq \overline{Q}_{ITS}-\frac{\log{\frac{1}{e}}}{\log{m}}\approx \frac{2m}{\mu}$, and does not depend on the value of $\alpha$. This is verified by our simulations shown in Figure \ref{fig:sim:1}, where we have simulated the attack with parameters $\mu=100$, $\epsilon=0.01$, and $\beta=0.1$. For each value $m= \{1000,2000,5000,10000\}$, we have simulated the attack $500$ times, by generating the ground-truth five times and choosing a victim randomly and uniformly for each generation $100$ times.

\subsection{Effect of Query Response Noise on $Q_{ITS}$}
Let us recall that the set $\Theta$ captures the diversity in query noise statistics due to the various hardware and software specifications of the users and the different browser sniffing techniques available to the attacker, where the resulting query response noise is captured by $P^{\theta(M)}_{Y|E_0}(\cdot|\cdot)$ and $M$ is the victim's index. 
Now, we investigate the effect of diversity of query response noise on the expected number of queries for successful deanonymization with ITS. To elaborate, we consider a noiseless scanned graph but noisy query responses. To model the query noise diversity, we consider two initial noise statistics $P_{Y|E_0}$ and $P'_{Y|E_0}$, where $P_{Y|E_0}$ is the transition probability of a binary symmetric channel with parameter $0.01$, and $P'_{Y|E_0}$ is the transition probability of a binary symmetric channel with parameter $0.3$. These statistics are chosen to be within the range of empirical observations of noise in browser history sniffing (e.g. \cite{gulmezoglu2017perfweb, smith2018browser,shusterman2019robust}). We consider 5 scenarios, where $\Theta_k=\{0,1,2,\cdots, 2^k-1\} ,k= [5]$, and define $P^{\theta}_{Y|E_0}= \frac{\theta}{2^k-1}P_{Y|E_0}+\frac{2^k-1-\theta}{2^k-1}P'_{Y|E_0},\theta\in \Theta_k$. Figure \ref{fig:sim:2} shows the resulting expected number of queries as a function of $m$, where we have simulated the attack with parameters $\alpha=1$, $\mu=100$, $\epsilon=0.1$, and $\beta=0.4$. For each value $m= \{1000,2000,5000,10000\}$, we have simulated the attack $500$ times, by generating the ground-truth five times and choosing a victim randomly and uniformly for each generation $100$ times.  It can be observed that increasing users' query noise diversity does not have a significant effect on the probability of success. 
\begin{figure}
 \centering \includegraphics[width=0.6\linewidth,height=2.5in, draft=false]{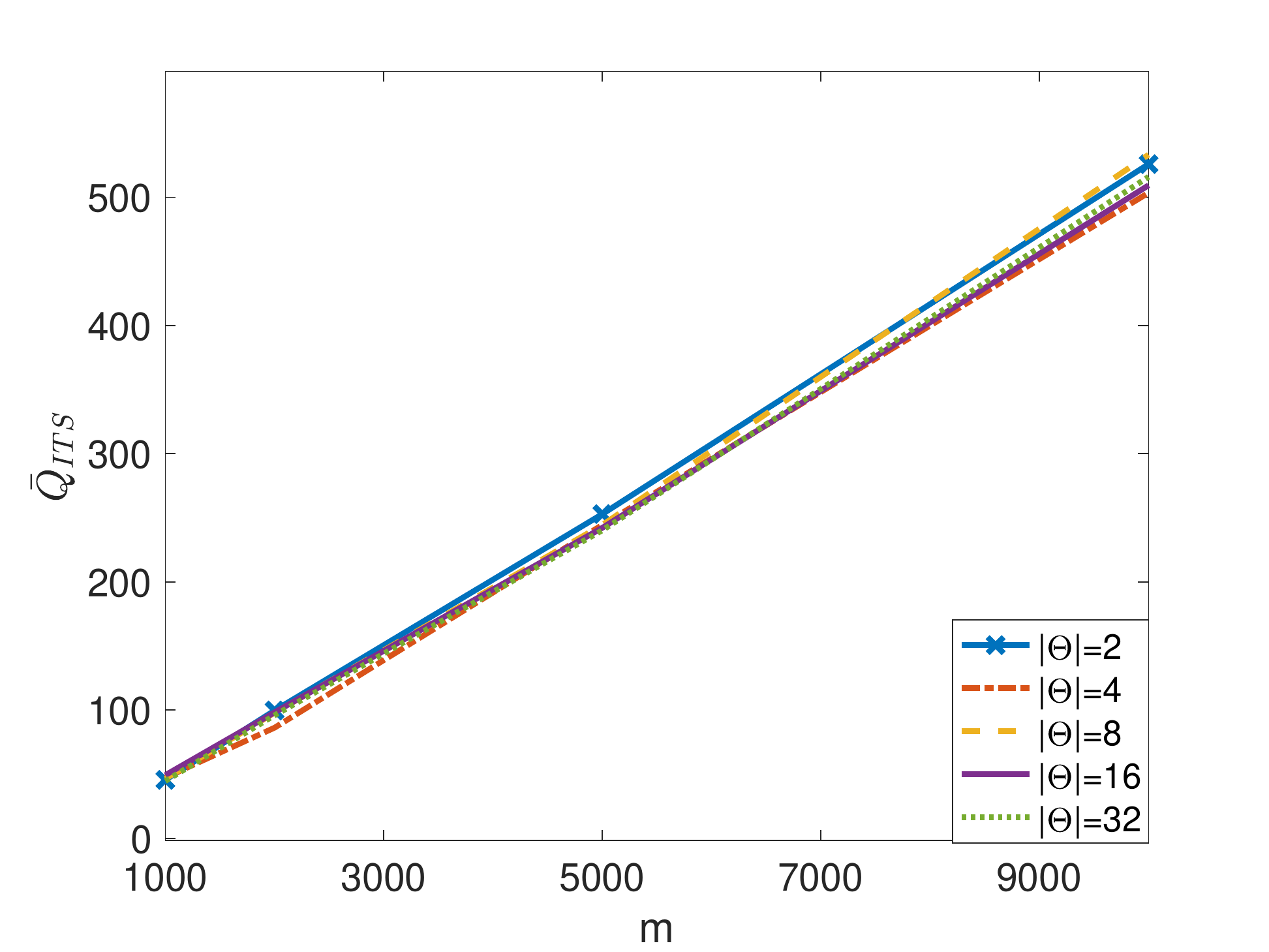}
  \caption{Expected number of queries $\overline{Q}$ necessary for success in the ITS under noiseless scan and noisy query response assumptions with a victim which is uniformly chosen among the users.}
 \label{fig:sim:2}
\end{figure}
\subsection{Effect of Scanned Noise on $\overline{Q}_{ITS}$}
In this section, we consider noiseless query responses, but noisy scanned graph and investigate the effects on the success of the ITS. As predicted by the theoretical results in Theorem \ref{th:3}, the diversity in scanned graph noise does not affect the expected number of queries. However, the upper-bound on the probability of error in Theorem \ref{th:3} changes linearly in $|\Gamma|$.  We have plotted the resulting probability of error in Figure \ref{fig:sim:3}, where
in order to model the scan noise diversity, we have considered two initial noise statistics $P_{E_s|E_0}$ and $P'_{E_s|E_0}$, where $P_{E_s|E_0}$ is the transition probability of a binary symmetric channel with parameter $0.01$ and $P'_{E_s|E_0}$ is the transition probability of a binary symmetric channel with parameter $0.3$.  We have considered five scenarios, where $k\triangleq|\Gamma|= \{2,4,8,16,32\}$, and defined $P^{\gamma}_{E_s|E_0}= \frac{\gamma-1}{k-1}P_{E_s|E_0}+\frac{k-\gamma}{k-1}P_{Y|E_0},\gamma\in [|\Gamma|]$. Figure \ref{fig:sim:3} shows the resulting probability of success ($1-P_e$) as a function of $m$, where we have simulated the attack with parameters $\alpha=1$, $\mu=100$, $\epsilon=0.1$, and $\beta=0.4$. For each value $m= \{1000,2000,5000,10000\}$, we have simulated the attack $500$ times, by generating the ground-truth five times and choosing a victim randomly and uniformly for each generation $100$ times. It can be observed that increasing the users' privacy preference options ($|\Gamma|$) does not have a significant effect on the resulting probability of success for ITS. This suggests that the upper-bound on the probability of error in Theorem \ref{th:3} can be potentially improved. 
\begin{figure}[t]
 \centering \includegraphics[width=0.6\linewidth,height=2.5in, draft=false]{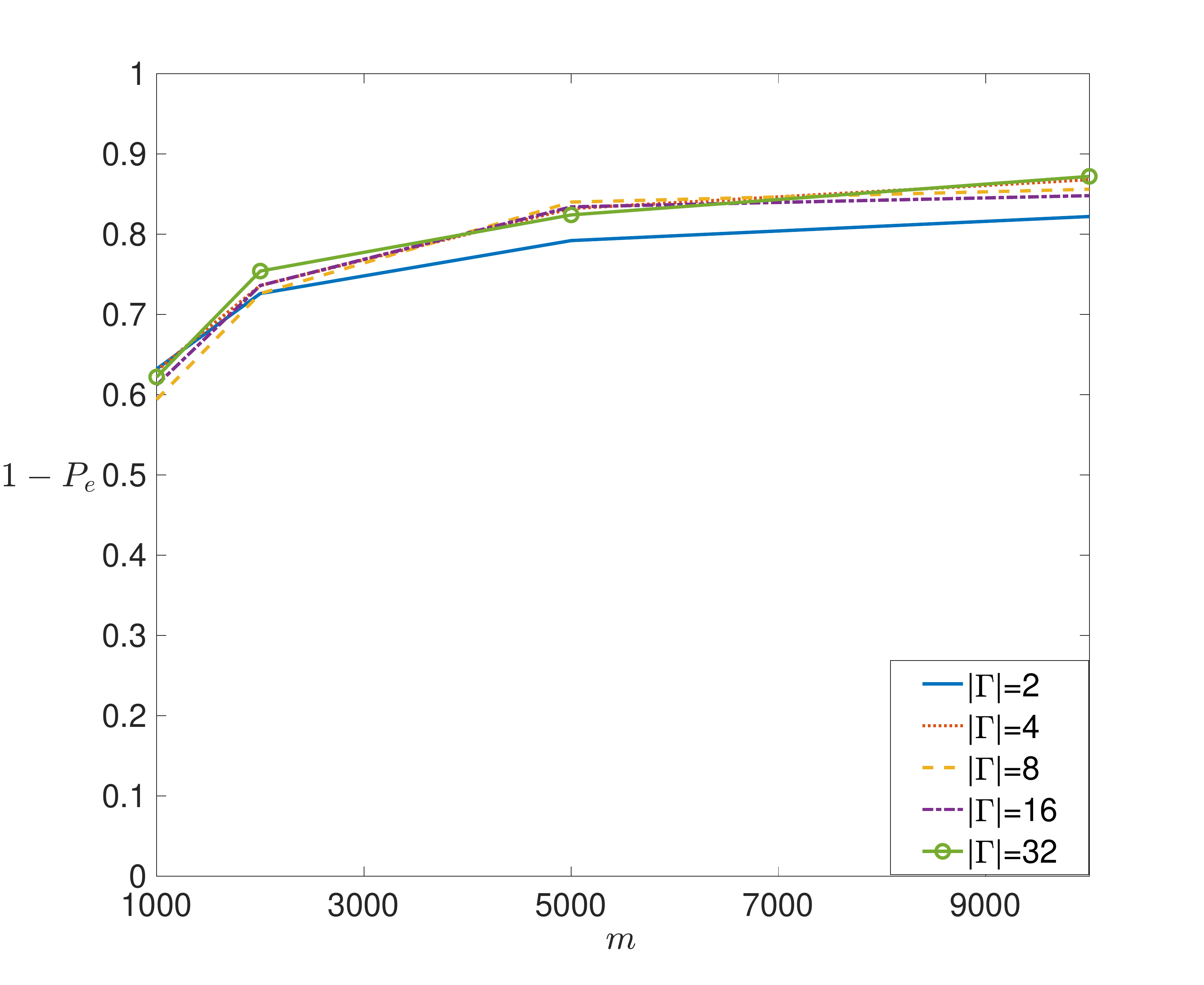}
  \caption{Probability of success in the ITS under varying scanned noise statistics and noiseless query response assumptions with a victim which is uniformly chosen among the users.}
 \label{fig:sim:3}
\end{figure}

\subsection{Performance in Real-world Networks}
In this section, we simulate an active attack on the LiveJournal network, which is a free on-line blogging community which allows users to form a group which other members can then join \cite{yang2015defining}. The database\footnote{The database is available at \href{https://snap.stanford.edu/data/com-LiveJournal.html}{https://snap.stanford.edu/data/com-LiveJournal.html}.} consists of $3,997,962$ members and $664,414$ groups. We have extracted a subset of $1517$ groups with at least $400$ members, and selected a subset of $49,164$ users which are  members of at least $4$ of these groups. The simulation is run $100$ times, where each time a victim is chosen randomly and uniformly among the users.  In Figure \ref{fig:sim:4}, we have simulated the attack in $10$ scenarios, where we have modeled both the scanning and query noise with binary erasure channels with erasure probability ranging from $ 0.01$ to $0.1$. We have used $\epsilon=0.1$ for the ITS error parameter. It can be seen that for the larger values of the erasure probability, the $1,517$ groups scanned by the attacker are not sufficient to identify the victim and the attacker must scan and query additional group memberships, whereas for smaller erasure probability, the attacker succeeds with probability close to one.

\begin{figure}[t]
 \centering 
\subfigure[]{\includegraphics[width=0.45\linewidth,height=2in, draft=false]{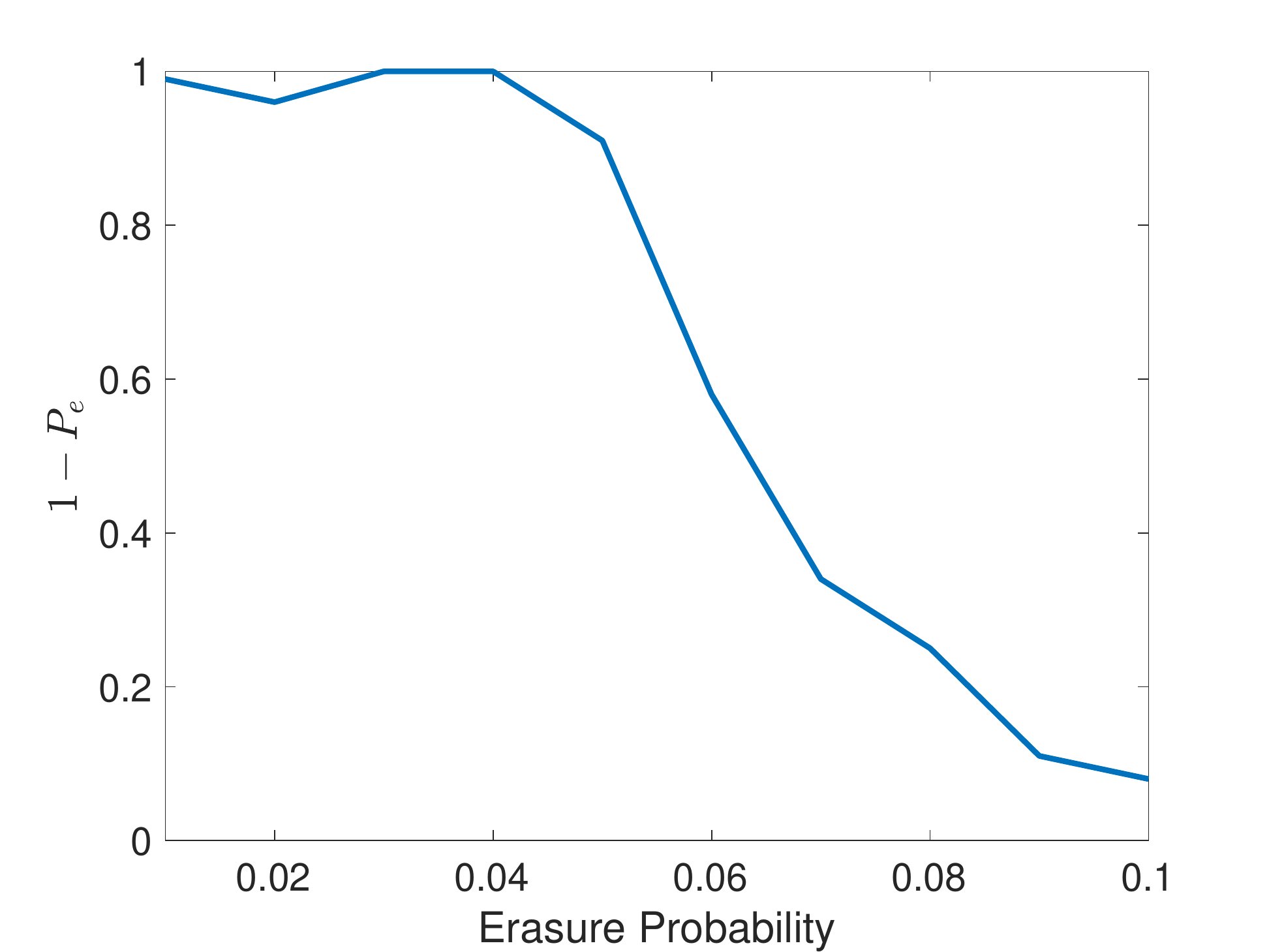}}
\subfigure[]{\includegraphics[width=0.45\linewidth,height=2in, draft=false]{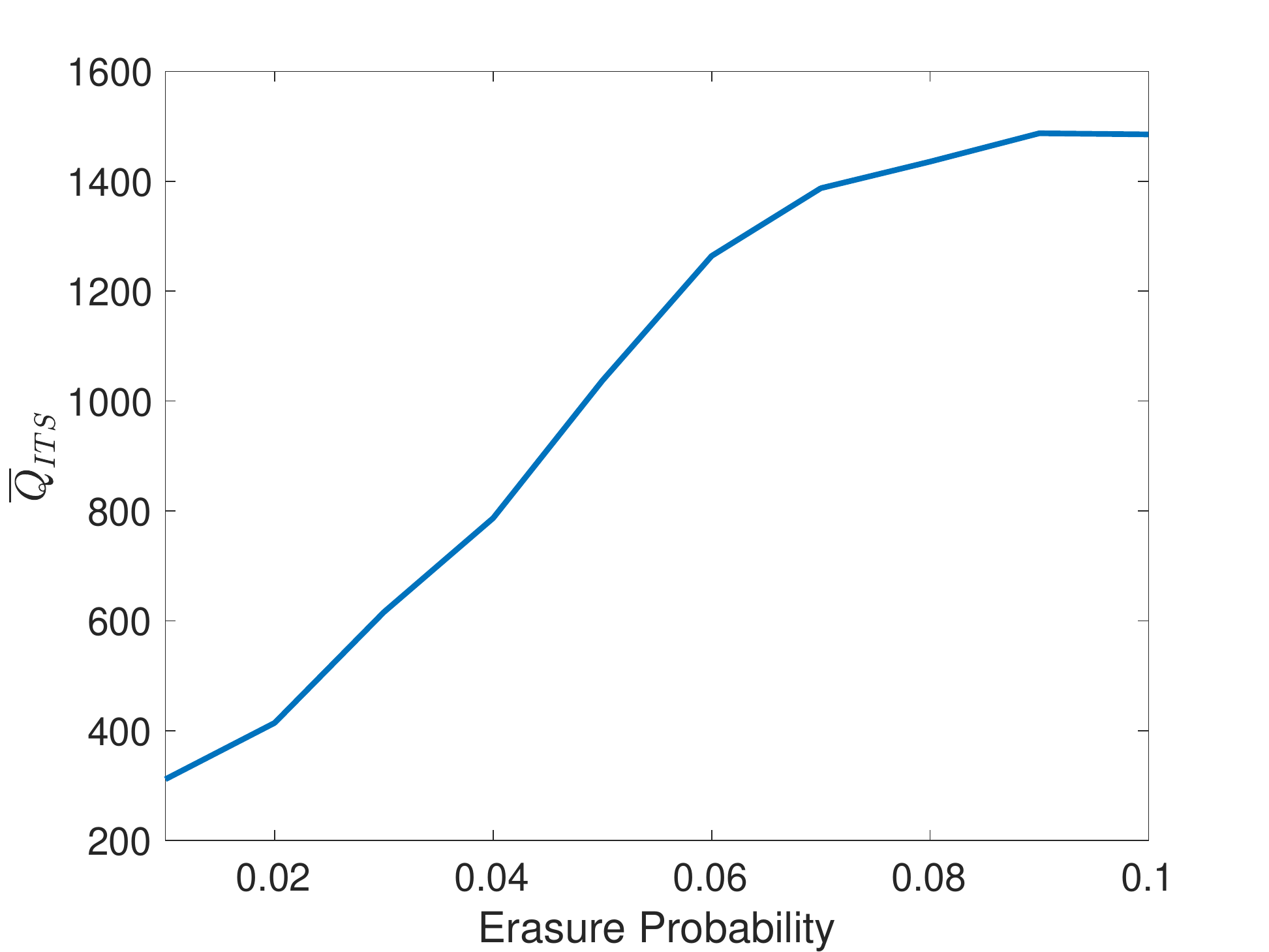}}
 \caption{ Probability of success (a) and average number of queries (b) for ITS in deanonymizing users in the LiveJournal network.}

 \label{fig:sim:4}
\end{figure}

\section{Conclusions and Future Work}
\label{sec:conc}
We have studied the ABND problem for general non-equiprobable user indices under various ground truth generation models such as linear and sublinear preferential attachment and stochastic block model. We have studied the ITS deanonymization strategy which operates based on information thresholds. The strategy measures the amount of uncertainty in the user indices given the received query responses. We have characterized the performance of the ITS both for social networks with a fixed, finite number of users as well as for asymptotically large social networks. We have provided simulations of the attack both in synthesized as well as real-world bipartite networks to verify the theoretical results. Future research directions include i) extending the theoretical results to scenarios where the scan and query noise models allow for correlated noise, ii) exploring the model assumptions such as sparsity in real-world bipartite networks other than social networks such as wireless mobility, and medical databases, and iii)
evaluating the performance of the proposed algorithms in such real-world bipartite networks.  

\begin{appendices}
\section{Proof of Proposition \ref{prop:1}}
\label{app:prop:1}
 To prove \eqref{eq:prop1:1}, note that, by symmetry, $\mathbb{E}(D_{\Delta})= \mathbb{E}(D_{\Delta,1})=\mathbb{E}(D_{\Delta,2})=\dots = \mathbb{E}(D_{\Delta,n})$. As a result, $\Delta= \sum_{j=1}^{n}\mathbb{E}(D_{\Delta,j})= n \mathbb{E}(D_{\Delta})$ and $\mathbb{E}(D_{\Delta})=\frac{\Delta}{n}=\mu$. 
 
 Next, we prove \eqref{eq:prop1:2}. Note that the group sizes $D_{t,j}, t\in [\Delta], j\in [n]$ are identically distributed since  in the $\alpha$-PA model the initial group popularities are assumed to be equal. More precisely, $P_{D_{t,j}}(d)= P_{D_{t,j'}}(d),d\in \{0,1,\cdots, m\}, j,j'\in [n], t\in [\Delta]$. As a result, we focus on $\mathbb{E}(D_t^2)$, where $D_t=D_{t,1}$. Let $N_t(d), t\in [\Delta], d\in \{0,1,\cdots, m\}$ be the number of groups with size equal to $d\in \{0,1,\cdots, m\}$ in step $t\in [\Delta]$ of the generation process. Note that
 \begin{align}
    \nonumber & \mathbb{E}(N_t(d))=\mathbb{E}(\sum_{j=1}^{n}\mathbbm{1}(D_{t,j}=d))
     = \sum_{j=1}^{n} P_{D_{t,j}}(d)\stackrel{(a)}{=}n P_{D_t}(d)
     \\&\label{eq:P_D}
     \Rightarrow P_{D_t}(d)= \frac{\mathbb{E}(N_t(d))}{n},
 \end{align}
 where in (a) we have used $P_{D_{t,j}}(d)= P_{D_t}(d), j\in [n], t\in [\Delta], d\in \{0,1,\cdots, m\}$.
 
 Also, we can write the conventional \emph{master equations} in growing networks \cite{dorogovtsev2000structure} using the law of total expectation: 
 \begin{align}
     &\mathbb{E}(N_t(0))= (1-p_{t-1}(0))  \mathbb{E}(N_{t-1}(0))\label{eq:master1}
     \\
     &\mathbb{E}(N_t(d))= (1-p_{t-1}(d))  \mathbb{E}(N_{t-1}(d))+\label{eq:master2}
     \\& p_{t-1}(d-1)\mathbb{E}(N_{t-1}(d-1)), 1<d\leq m \nonumber
 \end{align}
 where $p_t(d), t\in [\Delta], d\in \{0,1,\cdots, n\}$ is the probability that a given group $r_j$ with size $d$ is chosen in step $t$ of the generation process. By construction, we have $p_{t}(d) = \frac{d^{\alpha}+1}{\sum_{j=1}^n D^{\alpha}_{t,j}+n}\leq \frac{d^{\alpha}+1}{n}, t\in [\Delta], d\in \{0,1,\cdots, n-1\}$ and $p_t(m)=0$. 
 Note that:
 \begin{align*}
    & \mathbb{E}(D_t^2)=\sum_{d=0}^m d^2P_{D_t}(d)
\stackrel{(a)}{=} \sum_{d=1}^m d^2\frac{\mathbb{E}(N_t(d))}{n}
 \stackrel{(b)}{=}
    \sum_{d=1}^{m} \Bigg(\frac{d^2}{n}\bigg(\big(1-p_{t-1}(d)\big)  \mathbb{E}\big(N_{t-1}(d)\big)+
 p_{t-1}(d-1)\mathbb{E}\big(N_{t-1}(d-1)\big)\bigg )\Bigg)
\\
&   =\sum_{d=1}^{m} d^2\frac{\mathbb{E}\big(N_{t-1}(d))}{n}
     -
      \sum_{d=1}^{m} d^2\frac{\mathbb{E} (N_{t-1}(d))}{n}p_{t-1}(d) 
      +
\sum_{d=1}^{m}d^2 \frac{\mathbb{E}N_{t-1}(d-1)}{n}p_{t-1}(d-1)
\\
      &   =\sum_{d=1}^{m} d^2P_{D_{t-1}}(d)
     -
      \sum_{d=1}^{m} d^2P_{D_{t-1}}(d)p_{t-1}(d)
      +
 \sum_{d=1}^{m}d^2 P_{D_{t-1}}(d-1)p_{t-1}(d-1)
\\     &
      =\sum_{d=1}^{m} d^2P_{D_{t-1}}(d)
     -
      \sum_{d=1}^{m} d^2P_{D_{t-1}}(d)p_{t-1}(d)
      +
\sum_{d=0}^{m-1}(d+1)^2 P_{D_{t-1}}(d)p_{t-1}(d)
     \\     &
      =\sum_{d=1}^{m} d^2P_{D_{t-1}}(d)
     +
      \sum_{d=1}^{m} ((d+1)^2-d^2)P_{D_{t-1}}(d)p_{t-1}(d)
      -
      m^2P_{D_{t-1}}(m)p_{t-1}(m)
      \end{align*}
     \begin{align*}
     &\leq 
     \mathbb{E}(D^2_{t-1})
      +\sum_{d=0}^{m-1} (2d+1)P_{D_{t-1}}(d)p_{t-1}(d)
      \stackrel{(c)}{\leq} 
           \mathbb{E}(D^2_{t-1})
      +\sum_{d=0}^{m-1} (2d+1)P_{D_{t-1}}(d)\frac{d^{\alpha}+1}{n}
     \\&=    \mathbb{E}(D^2_{t-1})
      +\sum_{d=0}^{m-1} P_{D_{t-1}}(d)\frac{(2d^{1+\alpha}+d^{\alpha}+2d+1)}{n}
   = \mathbb{E}(D^2_{t-1})
      +\mathbb{E}(\frac{2D^{1+\alpha}_{t-1}+D^{\alpha}_{t-1}+2D_{t-1}+1}{n}) 
      \\& \stackrel{(d)}{\leq} 
      \mathbb{E}(D^2_{t-1})
      +\frac{2\mathbb{E}(D^2_{t-1})+3\mathbb{E}(D_{t-1})+1}{n},
 \end{align*}
 where in (a) we have used \eqref{eq:P_D}, in (b) we have used the master equations \eqref{eq:master1} and \eqref{eq:master2},  in (c) we have used the fact that $p_{t-1}(d) = \frac{d^{\alpha}+1}{\sum_{j=1}^n D^{\alpha}_{t-1,j}+n}\leq \frac{d^{\alpha}+1}{n}$, and in (d) we have used $\alpha\leq 1$.
 So,
 \begin{align}
     \mathbb{E}(D^2_t)\leq \mathbb{E}(D^2_{t-1})
     +\frac{2\mathbb{E}(D^2_{t-1})+3\mathbb{E}(D_{t-1}) +1}{n}, t\in [\Delta].
     \label{eq:fin_E}
 \end{align}
 Note that at the end of step $t=1$, there is exactly one edge in the graph, so $D_1$ is a Bernoulli variable with $P_{D_1}(0)=\frac{n-1}{n}$ and $P_{D_1}(1)=\frac{1}{n}$. So, $\mathbb{E}(D_1)=\mathbb{E}(D^2_1)=\frac{1}{n}$. Also, $D^2_t\geq D_t$ since $D_t\in\{0,1,\cdots,m\}$, so $\mathbb{E}(D^2_t)\geq \mathbb{E}(D_t)=\frac{t}{n}, t\geq 1$. As a result, from \eqref{eq:fin_E}, for any $t\in [\Delta]$ we have:
 \begin{align*}
   &     \mathbb{E}(D^2_t)\leq \mathbb{E}(D^2_{t-1})
     +\frac{5\mathbb{E}(D^2_{t-1})}{n}+\frac{1}{n}
     = \mathbb{E}(D^2_{t-1})(1+\frac{5}{n})+\frac{1}{n}
     \leq 
     ( \mathbb{E}(D^2_{t-2})((1+\frac{5}{n})+\frac{1}{n})(1+\frac{5}{n})+\frac{1}{n}
     \\&= \mathbb{E}(D^2_{t-2})(1+\frac{5}{n})^2+ \frac{1}{n}(1+(1+\frac{1}{n})).
 \end{align*}
 Consequently, as $n\to \infty$, we have:
 \begin{align*}
     &\mathbb{E}(D^2_\Delta)
     \leq (1+\frac{5}{n})^{\Delta-1}\mathbb{E}(D^2_1)+\frac{1}{n}(1+
     \sum_{i=2}^\Delta(1+\frac{5}{n})^{i-1} )
 \leq 
    \frac{e^{5\frac{\Delta}{n}}}{n}+
     \frac{1}{n}+e^{5\frac{\Delta}{n}}
\leq e^{5\mu}+O(\frac{1}{n})
     =O(1),
          \end{align*}
          where we have used $(1+\frac{a}{n})^n\to e^{a}$ as $n\to \infty$. 
Next, we prove \eqref{eq:prop1:3}. Take $i\neq j, i,j \in [n]$. We have:
\begin{align*}
    \mathbb{E}(D_{\Delta,i}D_{\Delta,j})
    &=  \mathbb{E}(D_{\Delta,1}D_{\Delta,2})
   =  \mathbb{E}(D_{\Delta,1}(\frac{1}{n-1}\sum_{j'=2}^{n}D_{\Delta,j'}))
    \\&= \mathbb{E}(D_{\Delta,1}(\frac{1}{n-1}\sum_{j'=1}^{n}D_{\Delta,j'}-D_{\Delta,1}))
= \mathbb{E}(D_{\Delta,1}(\frac{1}{n-1}\Delta-D_{\Delta,1}))
    \\&= \frac{\Delta}{n-1}\mathbb{E}(D_{\Delta,1})-\frac{1}{n-1}\mathbb{E}(D^2_{\Delta,1})
= \frac{\mu n}{n-1}\mathbb{E}(D_{\Delta,1})-\frac{1}{n-1}\mathbb{E}(D^2_{\Delta,1})
= \mu^2+O(\frac{1}{n}),
\end{align*}
where in the last inequality we have used  \eqref{eq:prop1:2}. The proof of \eqref{eq:prop1:2.5} follows by induction on the above argument.

Next, we prove \eqref{eq:prop1:4} using an inductive argument. For the basis of induction note that Similar to the proof of \eqref{eq:prop1:4}:
\begin{align*}
    \mathbb{E}(D^2_{\Delta,1}D_{\Delta,2})
    &= \frac{\Delta}{n-1}\mathbb{E}(D^2_{\Delta,1})-\frac{1}{n-1}\mathbb{E}(D^3_{\Delta,1})
\leq \frac{\Delta-1}{n-1 }\mathbb{E}(D^2_{\Delta,1})
     \leq \frac{\Delta}{n} \mathbb{E}(D^2_{\Delta,1})
    = \mu \mathbb{E}(D^2_{\Delta,1}),
\end{align*}
where we have used $\mu\geq 1$ to conclude that $\Delta\geq n$ and $\frac{\Delta-1}{n-1}\leq \frac{\Delta}{n}=\mu$. 
Furthermore, similar to the proof of \eqref{eq:prop1:4}, we have: 
\begin{align}
    &\label{eq:aux1}\mathbb{E}(D_{\Delta,1}^2D_{\Delta,2}D_{\Delta,3}\cdots D_{\Delta, \zeta})
= \frac{\Delta}{n-\zeta+1}\mathbb{E}(D_{\Delta,1}^2D_{\Delta,2}D_{\Delta,3}\cdots D_{\Delta, \zeta-1})-
    \\& \nonumber\frac{\zeta-2}{n-\zeta+1} \mathbb{E}(D_{\Delta,1}^2D^2_{\Delta,2}D_{\Delta,3}\cdots D_{\Delta, \zeta-1})-
    \nonumber\frac{1}{n-\zeta+1} \mathbb{E}(D_{\Delta,1}^3D_{\Delta,2}D_{\Delta,3}\cdots D_{\Delta, \zeta-1}).
\end{align}
Note that we have:
\begin{align*}
  &\mathbb{E}(D_{\Delta,1}^2D^2_{\Delta,2}D_{\Delta,3}\cdots D_{\Delta, \zeta-1})\geq  \mathbb{E}(D_{\Delta,1}^2D_{\Delta,2}D_{\Delta,3}\cdots D_{\Delta, \zeta-1}).
\end{align*}
Similarly, we have
$\mathbb{E}(D_{\Delta,1}^3D_{\Delta,2}D_{\Delta,3}\cdots D_{\Delta, \zeta-1})\geq  \mathbb{E}(D_{\Delta,1}^2D_{\Delta,2}$ $D_{\Delta,3}\cdots D_{\Delta, \zeta-1})$. So, from Equation \eqref{eq:aux1}:
\begin{align*}
    &\mathbb{E}(D_{\Delta,1}^2D_{\Delta,2}D_{\Delta,3}\cdots D_{\Delta, \zeta})\\
&\leq\frac{\Delta-\zeta+1}{n-\zeta+1}  \mathbb{E}(D_{\Delta,1}^2D_{\Delta,2}D_{\Delta,3}\cdots D_{\Delta, \zeta-1})
     \leq \frac{\Delta}{n}\mathbb{E}(D_{\Delta,1}^2D_{\Delta,2}D_{\Delta,3}\cdots D_{\Delta, \zeta-1})
    \\& = \mu \mathbb{E}(D_{\Delta,1}^2D_{\Delta,2}D_{\Delta,3}\cdots D_{\Delta, \zeta-1})
    \leq \mu^{\zeta-1} \mathbb{E}(D_{\Delta,1}^2),
\end{align*}
where the last inequality holds by induction. 

Lastly, we prove \eqref{eq:prop1:5}. Note that similar to the proof of \eqref{eq:prop1:3}:
\begin{align*}
    &\mathbb{E}(D_{\Delta,1}D_{\Delta,2}D_{\Delta,3}\cdots D_{\Delta, \zeta})= \frac{\Delta}{n-\zeta+1}\mathbb{E}(D_{\Delta,1}D_{\Delta,2}D_{\Delta,3}\cdots D_{\Delta, \zeta-1})- 
    \frac{\zeta-1}{n-\zeta+1}\mathbb{E}(D^2_{\Delta,1}D_{\Delta,2}D_{\Delta,3}\cdots D_{\Delta, \zeta-1})
    \\&\leq  \frac{\Delta}{n-\zeta+1}\mathbb{E}(D_{\Delta,1}D_{\Delta,2}D_{\Delta,3}\cdots D_{\Delta, \zeta-1})- 
    \frac{\zeta-1}{n-\zeta+1}\mathbb{E}(D_{\Delta,1}D_{\Delta,2}D_{\Delta,3}\cdots D_{\Delta, \zeta-1})
    \\&= \frac{\Delta-\zeta+1}{n-\zeta+1}\mathbb{E}(D_{\Delta,1}D_{\Delta,2}D_{\Delta,3}\cdots D_{\Delta, \zeta-1})
  \leq \mu\mathbb{E}(D_{\Delta,1}D_{\Delta,2}D_{\Delta,3}\cdots D_{\Delta, \zeta-1}),
\end{align*}
where in the last inequality we have used $\Delta>n$ and $\frac{\Delta}{n}=\mu$. The rest of the proof follows by induction. 

\section{Proof of Proposition \ref{prop:2}}
\label{app:prop:2}
 Note that $\frac{C_i}{n}= \frac{1}{n}\sum_{j=1}^n \mathbbm{1}(R(i,j))$, where $R_{i,j}=\mathbbm{1}(u_i\in \mathcal{U}_j), i\in [m], j\in [n]$. Also,  for any $\mathcal{A}\subset [n]$ we have:
\begin{align*}
    &\mathbb{E}((R_{i,j})_{j\in \mathcal{A}})= P(R_{i,j}=1, j\in \mathcal{A})
    = \sum_{d^{\mathcal{A}}\in \{0,1,\cdots,m\}^{|\mathcal{A}|}}P_{(D_j)_{j\in \mathcal{A}}}(d^{\mathcal{A}})P(R_{i,j}=1, j\in \mathcal{A}|d^{\mathcal{A}})
\\&\stackrel{(a)}{=} \sum_{d^{\mathcal{A}}\in \{0,1,\cdots,m\}^{|\mathcal{A}|}}P_{(D_j)_{j\in \mathcal{A}}}(d^{\mathcal{A}})\prod_{j\in \mathcal{A}} P(R_{i,j}=1|d_j)
= \sum_{d^{\mathcal{A}}\in \{0,1,\cdots, m\}^{|\mathcal{A}|}}P_{(D_j)_{j\in \mathcal{A}}}(d^{\mathcal{A}})\prod_{j\in \mathcal{A}} \frac{d_j}{m}
{=} \frac{1}{m^{|\mathcal{A}|}}\mathbb{E}(\prod_{j\in \mathcal{A}}D_j)
    \stackrel{(b)}{\leq} \frac{\mu}{m}^{|\mathcal{A}|}.
\end{align*}
where in (a) we have used the fact that given the group sizes, the users' memberships in the groups are independent of each other by construction. The reason is that in the graph generation process, at each step, once a group is chosen, a user is chosen randomly and uniformly and added to that group's members independent of the previous members, and (b) follows from \eqref{eq:prop1:5}.

So, using an extension of Hoeffding's inequality to weakly correlated variables given in Theorem 3 in \cite{impagliazzo2010constructive}, we have:
\begin{align*}
   & P(C_i\geq \ell)\leq c2^{-nD_b(\frac{\mu}{m}(1+\psi)||\frac{\mu}{m})},
\end{align*}
where $\ell= \frac{n}{m}\mu(1+\psi)=\frac{1}{\beta}\mu(1+\epsilon)$ and $\psi\in (0, \frac{m}{\mu}-1)$. To derive \eqref{eq:prop2}, we note that: 
\begin{align*}
   & P(C_i\geq \ell)\leq c2^{-n(\frac{\mu}{m}(1+\psi)\log{(1+\psi)}+O(\frac{1}{n}))}\to 0,\text{ as } n\to \infty.
\end{align*}


\section{Proof of Proposition \ref{prop:3}}
\label{app:prop:3}
We prove the statement by induction. The case of $n'=1$ is trivially true. Let $n'>1$ and assume that the statement is true for all $n''<n'$.
To simplify the notation, and without loss of generality, we assume that $j_k=k, k\in [n']$. 
Note that given the group sizes, the users' memberships in the groups are independent of each other by construction.  As a result, given that  $D_{k}=d_{k}$, each user is a member of $s_{k}$ with probability $\frac{d_{k}}{m}$ independent of all other users. So, we use the law of total probability and condition on the group sizes:
\begin{align*}
     &P((R_{k})_{k\in [n']}=s^{n'})=
     \sum_{d^{n'}\in \{0,1,\cdots, m\}^{n'}} P((D_{k})_{k\in [n']}=d^{n'})  P((R_{k})_{k\in [n']}=s^{n'}|(D_{k})_{k\in [n']}=d^{n'} )
     \end{align*}
     We have:
     \begin{align*}
   &P((R_{k})_{k\in [n']}=s^{n'}|(D_{k})_{k\in [n']}=d^{n'} )=
=\prod_{k\in [n']}P(R_{k}=s_k|D_{k}=d_k)
    = \prod_{k\in [n']}g(\frac{d_{k}}{m},s_i),
\end{align*}
where 
\begin{align*}
    g(\frac{d_{k}}{m},s_k)= 
    \begin{cases}
    \frac{d_{k}}{m}\qquad& \text{ if } s_k=1,
\\    1-\frac{d_{k}}{m}& \text{ if } s_k=0
    \end{cases}.
\end{align*}
So, 
\begin{align*}
     &P((R_{k})_{k\in [n']}=s^{n'})=
     \sum_{d^{n'}\in \{0,1,\cdots,m\}^{n'}} P((D_{k})_{k\in [n']}=d^{n'})   \prod_{k\in [n']}g(\frac{d_{k}}{m},s_k)
     = \mathbb{E}(\prod_{k\in [n']}g(\frac{d_{k}}{m},s_i)).
\end{align*}
Let us assume that $\sum_{k=1}^{n'}s_k= \zeta, \zeta\in [1,n]$. Using symmetry and
without loss of generality, let the first $\zeta$ elements $s_k, k\in [\zeta]$ be equal to $1$ and the rest of are equal to $0$, so that  $s_k=\mathbbm{1}(k\leq \zeta), j \in [n']$. Then, 
\begin{align}
\label{eq:prop3:aux}
     &P((R_{k})_{k\in [n']}=s^{n'})
     = \frac{1}{m^{n'}}\mathbb{E}(\prod_{k\in [\zeta]}D_{k}\prod_{k'\in [\zeta+1,n']}(m-D_{k'})).
\end{align}
We derive lower and upper bounds for the right hand side of the last equality to complete the proof. To derive the upper bound, note that:
\begin{align*}
  & \frac{1}{m^{n'}}\mathbb{E}(\prod_{k\in [\zeta]}D_k\prod_{k'\in [\zeta+1,n']}(m-D_{k'}))\leq
   \frac{1}{m^{n'}}\mathbb{E}(m^{n'-\zeta}\prod_{k\in [\zeta]}D_k)
\leq \frac{m^{n'-\zeta}}{m^{n'}}\mu^{\zeta}= (\frac{\mu}{m})^{\zeta},
\end{align*}
where the last inequality follows from \eqref{eq:prop1:5}. Consequently, 
\begin{align*}
  & \frac{1}{m^{n'}}\mathbb{E}(\prod_{k\in [\zeta]}(m-D_{1})\prod_{k'\in [\zeta+1,n']}D_j)
 \stackrel{(a)}\leq e^{\frac{\mu}{\beta}}(1-\frac{\mu}{m})^{n'-\zeta}(\frac{\mu}{m})^{\zeta},
\end{align*}
where (a) follows from:
\begin{align*}
    e^{-\frac{\mu}{\beta}}\leq (1-\frac{\mu}{m})^{n'-\zeta}\leq 1,
\end{align*}
which holds by taking $n'\to n$ and $n\to \infty$ and noting that $\lim_{n\to \infty} (1-\frac{a}{n})^n\to e^{-a}$ as $n\to \infty$.

To derive a lower bound for 
\eqref{eq:prop3:aux}, note that:
\begin{align*}
    &\frac{1}{m^{n'}}\mathbb{E}(\prod_{k\in [\zeta]}D_{k}\prod_{k'\in [\zeta+1,n']}(m-D_{k'}))=
    \\&
    \frac{1}{m^{n'}}\bigg( m\mathbb{E}(\prod_{k\in [\zeta]}D_{k}\prod_{k'\in [\zeta+1,n'-1]}(m-D_{k'}))-
\mathbb{E}(D_{n'}\prod_{k\in [\zeta]}D_{k}\prod_{k'\in [\zeta+1,n'-1]}(m-D_{k'}))\bigg)
    \\&\stackrel{(a)}{\geq} 
      \frac{1}{m^{n'-1}} \mathbb{E}(\prod_{k\in [\zeta]}D_{k}\prod_{k'\in [\zeta+1,n'-1]}(m-D_{k'}))-
    (\frac{\mu}{m})^{\zeta+1}
    \\&\stackrel{(b)}{\geq} 
    \frac{1}{m^{n'-1}} \mathbb{E}(\prod_{k\in [\zeta]}D_{k}\prod_{k'\in [\zeta+1,n'-1]}(m-D_{k'}))-
   \frac{\mu}{m} \prod_{k=1}^{n'} P_R(s_k),
\end{align*}
 where (a) and (b) follow from similar arguments as the ones in the derivation of the upper bound. Also, by the induction assumption, $\frac{1}{m^{n'-1}} \mathbb{E}(\prod_{k\in [\zeta]}D_{k}\prod_{k'\in [\zeta+1,n'-1]}(m-D_{k'}))\geq (1-\frac{(n'-1)\mu}{m})\prod_{n=1}^{n'-1}P_{R}(s_k)$.  So, 
 \begin{align*}
   & \frac{1}{m^{n'}}\mathbb{E}(\prod_{k\in [\zeta]}D_{k}\prod_{k'\in [\zeta+1,n']}(m-D_{k'}))
   \geq  
    (1-\frac{(n'-1)\mu}{m})\prod_{n=1}^{n'-1}P_{R}(s_k)-  \frac{\mu}{m} \prod_{k=1}^{n'} P_R(s_k)
    \\&
    \geq \prod_{n=1}^{n'}P_{R}(s_k)( (1-\frac{(n'-1)\mu}{m}-  \frac{\mu}{m})
    \geq (1-\frac{n'\mu}{m})\prod_{n=1}^{n'-1}P_{R}(s_k),
 \end{align*}
The upper bound for 
\eqref{eq:prop3:aux} when $n'\geq \frac{\mu}{m}$ and the sparsity condition $\sum_{i=1}^{n'}\mathbbm{1}(s_i=1)\leq C$ follows similar steps as the derivation above and the following argument:
\begin{align*}
  & \frac{1}{m^{n'}}\mathbb{E}(\prod_{k\in [\zeta]}D_k\prod_{k'\in [\zeta+1,n']}(m-D_{k'}))\leq
   \frac{1}{m^{n'}}\mathbb{E}(m^{n'-\zeta}\prod_{k\in [\zeta]}D_k)
\leq \frac{(m-\mu)^{n'-\zeta}}{m^{n'}}\mu^{\zeta}(1+o(1))
   \\&=(1-\frac{\mu}{m})^{n'-\zeta} (\frac{\mu}{m})^{\zeta}(1+o(1)).
\end{align*}

Next, we derive a lower bound for 
\eqref{eq:prop3:aux} when $n'\geq \frac{\mu}{m}$ and under the the sparsity condition $\sum_{i=1}^{n'}\mathbbm{1}(s_i=1)\leq C$ for some constant finite number $C>0$. Note that:

\begin{align}
 \nonumber      &\prod_{k=1}^{n'}P_{R}(s_k)= 
    (1-\frac{\mu}{m})^{n'-\zeta}(\frac{\mu}{m})^{\zeta}
    \\&\stackrel{(a)}{\leq}2e^{-\frac{\mu(n'-\zeta)}{m}}(\frac{\mu}{m})^{\zeta}
    = 2 e^{-\frac{\mu n'}{m}}(\frac{\mu}{m})^{\zeta}(1+o(1))
   \\&\label{eq:up} \stackrel{(b)}{\leq} 2e^{-\frac{ n}{m}}(\frac{\mu}{m})^{\zeta}(1+O(\frac{1}{n}))
   =2e^{-\frac{1 }{\beta}}(\frac{\mu}{m})^{\zeta}(1+o(1))),
\end{align}
where in (a), we have used the fact that $\zeta=o(n)$ and that $n'>\frac{n}{\mu}$ to conclude that the inequality holds as $n\to \infty$, and in (b) we have used $n'\geq \frac{k_1 n}{\mu}$. On the other hand:
\begin{align}
\nonumber    &\frac{1}{m^{n'}}\mathbb{E}(\prod_{k\in [\zeta]}D_{k}\prod_{k'\in [\zeta+1,n']}(m-D_{k'}))
    \stackrel{(a)}\geq 
    \frac{1}{m^{n'}}( \mathbb{E}(\prod_{k\in [\zeta]}D_{k}\prod_{k'\in [\zeta+1,n']}m^{n'-\zeta-\mu}(m-n)^{\mu})
   = \frac{1}{m^{\zeta}} (1-\frac{1}{\beta})^{\mu}
    \mathbb{E}(\prod_{k\in [\zeta]}D_{k})
    \\&\label{eq:low} \stackrel{(b)}{=} 
  (1-\frac{1}{\beta})^{\mu}
    (\frac{\mu}{m})^{\zeta}(1+o(1)),
\end{align}
where in (a) we have used the fact that $\sum_{k'=\zeta+1}^{n'}D_{k'}\leq \Delta$ to conclude that $\prod_{k'\in [\zeta+1,n' ]} (m-D_{k'})\geq m^{n'-\zeta-\mu}(m-n)^{\mu}$, and in (b) we have used \eqref{eq:prop1:2.5}.
Combining \eqref{eq:low} and \eqref{eq:up} completes the proof. 

\section{Proof of Theorem \ref{th:1}}
\label{app:th:1}
 The proof builds upon ideas developed for studying the fundamental limit of communication over feedback channels \cite{ burnashev1976data}.
Note that $I(E_0;E_s)\leq H(E_0)= (\frac{\mu}{m}\log{m})(1+o(\frac{\log{m}}{m}))$ since $P(E_0=1)=\frac{\mu}{m}$. So, the upper-bound on $\overline{Q}$ is greater than $C_3\frac{n}{\mu}$ for some constant $C_3>0$. This allows us to use Proposition \ref{prop:3} to approximate the fingerprint distribution by a memoryless distribution as a lower bound. 
Define the following stopping times
\begin{align*}
    &\kappa_k\triangleq\min_{\kappa>\frac{C_3 n}{\mu}}\bigg\{\kappa\big|  I_{k}(\mathcal{G}_s,Y^\kappa)>
    \log{\frac{1}{\epsilon}}\bigg\}, k\in [m],
   \qquad \qquad \kappa^*\triangleq\min_{k\in [m]}\kappa_m
\end{align*}
Note that by the definition of the identification function $Id_t(\cdot,\cdot)$ in \eqref{eq:id}, we have $\overline{Q}_{ITS}= \mathbb{E}(\kappa^*)$. 
We show that $\mathbb{E}(\kappa^*)\leq \frac{H(M)+\log{\frac{1}{\epsilon}}+ i_{\max}}{I(E_0;E_s)}$.
Note that $\mathbb{E}(\kappa^*)\leq \mathbb{E}(\kappa_M)$ by definition of $\kappa^*$. So, it is enough to prove the upper bound on $\mathbb{E}(\kappa_M)$. Fix $n'\in \mathbb{N}$. Let $T_{n'}= \min\{\kappa_M,n'\}$. Note that:
 \begin{align}
\nonumber  \mathbb{E}\left(I_{T_{n'}}\left(M\right)\right)&
= \mathbb{E}\left(\sum_{i=1}^{T_{n'}}    \log{\frac{P_{E_0|E_s}(Y_i|F_{k,i})}{P_{E_0}({Y_t})}}+I_o(J)\right)
\stackrel{(a)}{=}\mathbb{E}\left(\sum_{j=1}^{T_{n'}}\log{\frac{P_{E_0|E_s}(Y_i|F_{k,i})}{P_{E_0}({Y_t})}}\right)-H\left(J\right)
\\& \label{eq:th1:1}\stackrel{(b)}{\geq}c'I(E_0;E_s)\mathbb{E}(T_{n'})-H(J),
\end{align}
where (a) uses linearity of expectation, and (b) follows from Wald's identity \cite{wald1944cumulative} and using $P((F_{k,i})_{i\in [t]},Y^t )\leq c' \prod_{i=1}^t P_{E_0,E_s}(F_{k,i},Y_i), t\in \mathbb{N}$ which holds due to Proposition \ref{prop:3}, and the fact that 
\begin{align*}
   & P((F_{k,i})_{i\in [t]},Y^t )=
 \sum_{s^{t}\in \{0,1\}^t} P((R_{k,i})_{i\in {t}}=s^t)\prod_{i\in [t]}P_{E_{s}|E_0}
 (F_{k,i}|s_i)P_{Y|E_0}(Y_i|s_i)
 \\&=  \sum_{s^{t}\in \{0,1\}^t} P((R_{k,i})_{i\in {t}}=s^t)\prod_{i\in [t]}P_{E_{s}|E_0}
 (F_{k,i}|s_i)\mathbbm{1}(Y_i=s_i)
 \leq c' \prod_{i\in [t]}\sum_{s_i\in \{0,1\}} P(R_{k,i}=s_i)P_{E_{s}|E_0}
 (F_{k,i}|s_i)\mathbbm{1}(Y_i=s_i)
\\&= c'  \prod_{i\in [t]}P_{E_0,E_s}(F_{k,i},Y_i). 
\end{align*}
Note that the sparsity condition $\sum_{i\in [t]} s_i= o(n)$ in Proposition \ref{prop:3} is satisfied with probability one due to Proposition \ref{prop:2}.
Also, note that at each step $t\in \mathbb{N}$, the increase in $I_t(M)$ is less than or equal to $i_{max}$. It follows that:
\begin{align}
\label{eq:th1:2}
 \mathbb{E}\left(I_{T_{n'}}\left(M\right)\right)\leq \mathbb{E}\left(I_{T_{n'}-1}\left(M\right)\right)+i_{max}\leq \log\frac{1}{\epsilon}+i_{max},
\end{align}
where we have used the fact that and that by the definition of $\kappa_M$ and $T_{n'}$, we have $I_{T_{n'}-1}(M)\leq \log\frac{1}{\epsilon}$ since $T_{n'}-1<\kappa_M$. Combining \eqref{eq:th1:1} and \eqref{eq:th1:2} we get $\mathbb{E}(T_{n'})\leq \frac{H(M)+\log{\frac{1}{\epsilon}}+ i_{\max}}{c'I(E_0;E_s)}$, and using the monotone convergence theorem by increasing $n'$ asymptotically, we get $\mathbb{E}(T_{n'})=\mathbb{E}(\kappa_M)=\overline{Q}_{ITS}$ which yields the desired bound on $\overline{Q}_{ITS}$. It remains to prove the bound on $P_{e,ITS}$. We have:
\begin{align*}
  &  P_e= P(\exists j\neq M: \kappa_j\leq \kappa_M)\leq \sum_{j\neq M} P(\kappa_j\leq \infty)
  = \sum_{j\neq M} \lim_{\eta\to \infty} P(\kappa_j\leq \eta)
  \\&\stackrel{(a)}{=}\sum_{j\neq M} \lim_{\eta\to \infty} \mathbb{E}_{P_{Y^n,(F_{M,i})_{i\in [n]}}}\left(\frac{P_{Y^n}P_{(F_{M,i})_{i\in [n]}}}{P_{Y^n,(F_{M,i})_{i\in [n]}}}\mathbbm{1}(\kappa_j\leq \eta))\right)
  \\&\leq \sum_{j\neq M} \lim_{\eta\to \infty} 
  \frac{(1+o(1))}{c'}\mathbb{E}_{P_{Y_i,F_{M,i}}}\left(\prod_{i\in [n]}\frac{P_{Y_i}P_{F_{M,i}}}{P_{Y_i,F_{M,i}}}\mathbbm{1}(\kappa_j\leq \eta))\right)
  \\&= \sum_{j\neq M} \lim_{\eta\to \infty} 
  \frac{(1+o(1))}{c'}\mathbb{E}_{P_{Y_i,F_{M,i}}}\left(\prod_{i\in [\eta]}\frac{P_{Y_i}P_{F_{M,i}}}{P_{Y_i,F_{M,i}}}\mathbbm{1}(\kappa_j\leq \eta))\right)\times
  {E}_{P_{Y_i,F_{M,i}}}\left(\prod_{i\in [\eta+1,n]}\frac{P_{Y_i}P_{F_{M,i}}}{P_{Y_i,F_{M,i}}}\right)
\\
  & \sum_{j\neq M} \lim_{\eta\to \infty} 
  \frac{(1+o(1))}{c'}\mathbb{E}_{P_{Y_i,F_{M,i}}}\left(\prod_{i\in [\eta]}\frac{P_{Y_i}P_{F_{M,i}}}{P_{Y_i,F_{M,i}}}\mathbbm{1}(\kappa_j\leq \eta))\right)
  \\&= \sum_{j\neq M} \lim_{\eta\to \infty} 
  \frac{(1+o(1))}{c'}\mathbb{E}_{P_{Y_i,F_{M,i}}}\left(e^{\sum_{i\in [\eta]}\log\frac{P_{Y_i}P_{F_{M,i}}}{P_{Y_i,F_{M,i}}}}\mathbbm{1}(\kappa_j\leq \eta))\right)
  \end{align*}
  \begin{align*}
  &{\leq} 
  \sum_{j\neq M} \lim_{\eta\to \infty} 
  \frac{(1+o(1))}{c'}\mathbb{E}_{P_{Y_i,F_{M,i}}}\left(e^{\sum_{i\in [\eta]}\log\frac{P_{Y_i}P_{F_{M,i}}}{P_{Y_i,F_{M,i}}}})\right)
= 
  \sum_{j\neq M} \lim_{\eta\to \infty} 
  \frac{(1+o(1))}{c'}\mathbb{E}_{P_{Y_i,F_{M,i}}}\left(e^{-I_{\eta}(M)- I_0(M)})\right)
  \\&\leq 
  \sum_{j\neq M} \lim_{\eta\to \infty} 
  \frac{(1+o(1))}{c'}\mathbb{E}_{P_{Y_i,F_{M,i}}}\left(e^{-\log\frac{1}{\epsilon}-I_0(M)})\right)
  =  \sum_{j\neq M} \frac{1}{c'}\epsilon P_M(j)\leq \frac{1}{c'}\epsilon(1+o(1)).
\end{align*}
where in (a) we have used the fact that $P_{(F_{j,i})_{i\in [n]}}= P_{(F_{M,i})_{i\in [n]}}, j\in [m]$. 

\end{appendices}

\bibliographystyle{unsrt}

\end{document}